\documentclass[conference]{IEEEtran}
\IEEEoverridecommandlockouts
\usepackage{cite}
\usepackage{amsmath,amssymb,amsfonts}
\usepackage{amsthm}
\usepackage{textcomp}
\usepackage{xcolor}
\usepackage{multirow}
\usepackage{algorithm}
\usepackage{algpseudocode}
\usepackage{graphicx}
\usepackage{caption}
\usepackage{subcaption}
\usepackage{booktabs}
\usepackage{mathrsfs}
\newtheorem{lemma}{Lemma}
\newtheorem{proposition}[lemma]{Proposition}
\newtheorem{corollary}[lemma]{Corollary}
\newtheorem{definition}[lemma]{Definition}
\newtheorem{theorem}[lemma]{Theorem}
\definecolor{mplgreen}{rgb}{0.12, 0.50, 0.12}
\definecolor{mplblue}{rgb}{0.15, 0.42, 0.65}
\usepackage[colorlinks=true, urlcolor=blue, linkcolor=blue, citecolor=red]{hyperref}
\def\BibTeX{{\rm B\kern-.05em{\sc i\kern-.025em b}\kern-.08em
    T\kern-.1667em\lower.7ex\hbox{E}\kern-.125emX}}
\newenvironment{CompactItemize}{
\begin{list}{\tiny$\bullet$}{%
\setlength{\leftmargin}{10pt}
\setlength{\itemindent}{0pt}
\setlength{\topsep}{1pt}
\setlength{\itemsep}{1pt}
}}
{\end{list}}

\begin{document}

\title{QASP: Query-Adaptive Robust Vector Search Policy}

\author{\IEEEauthorblockN{Hakan Ferhatosmanoglu$^*$}
\IEEEauthorblockA{
\textit{Amazon}\\
London, UK \\
hakanf@amazon.com}
\and
\IEEEauthorblockN{Kushal Kumar$^*$}
\IEEEauthorblockA{
\textit{Amazon}\\
New York, USA \\
kushlku@amazon.com}
\and
\IEEEauthorblockN{Tal Wagner}
\IEEEauthorblockA{
\textit{Amazon and Tel-Aviv University}\\
Tel Aviv, Israel \\
talw@amazon.com}
\and
\IEEEauthorblockN{Andy Warfield}
\IEEEauthorblockA{
\textit{Amazon}\\
Vancouver, Canada \\
warfield@amazon.com}
\thanks{$^*$Equal contribution.}
}

\maketitle

\begin{abstract}
A fundamental challenge of vector search is achieving consistently high recall while minimizing computational costs. Fixed search parameters cause significant performance variance across queries, and conventional evaluation on average recall masks these per-query disparities. We introduce QASP (Query-Adaptive robust vector Search Policy), which predicts the complete recall progression curve per query via a single upfront supervised regression, from which a search policy is derived for any recall target; this avoids iterative model invocations during search or separate predictors per target. By predicting normalized recall values with scale-invariant features and pre-search inference, QASP generalizes across recall targets, index configurations, and datasets. The fine-grained progress predictions further enable a lightweight reactive complement that adjusts search depth based on predicted-versus-observed deviations without additional inference. 
We prove that QASP requires a finite training sample independent of dataset size and dimensionality, that it matches or improves on the best fixed-policy loss up to a vanishing statistical error, and that its data access savings over fixed probing grow exponentially in intrinsic dimensionality. Experimentally, QASP achieves significantly lower recall variance and deviation from target, higher query satisfaction rate, and scales to large data and hierarchical indices without retraining, achieving 99\% recall with 80\% less data access. Our implementation can be found at \url{https://github.com/amazon-science/QASP}.
\end{abstract}


\section{Introduction}
Given a query vector $q \in \mathbb{R}^d$ and a dataset $X \subset \mathbb{R}^d$, vector (similarity) search aims to find the $k$ vectors $\{x^q_1, \dots, x^q_k\} \subset X$ that minimize $\text{dist}(q, x^q_{\cdot})$, where $\text{dist}(\cdot,\cdot)$ is a notion of distance.  
The fraction of $k$ closest vectors found is the \emph{recall} for that query. A vector search policy determines the search parameters to achieve high recall with low cost.

Practitioners typically rely on heuristic guidelines to set search parameters based on average performance, generating rules such as accessing a fixed proportion $\alpha \in [0.05, 0.1]$ of the dataset or $\lfloor \sqrt{L} \rfloor$ out of $L$ partitions \cite{lancedb_guide, faiss_guidelines, opensearch_guide}. Query-agnostic settings fail to capture heterogeneous query difficulty and the resulting variation in computational needs \cite{laet, quake, 11112943, 11113056}. This causes over-reading for easy queries and under-reading for difficult ones. Conventional evaluation methods compound this problem by focusing on average recall, which masks per-query disparities. Figure \ref{fig:motivation} illustrates how query-agnostic policies access significantly more data than the query-adaptive approach. 
Such fixed policies may require substantial over-reading to achieve the consistently high recall required by production applications, including semantic search, retrieval-augmented generation, recommendation systems, and entity resolution.

Our goal is to design a proactive search policy that consistently achieves high recall with minimal performance variance across queries regardless of difficulty. We present QASP, \textbf{Q}uery-\textbf{A}daptive robust vector \textbf{S}earch \textbf{P}olicy, which addresses this through supervised regression that predicts the complete recall progression curve, from which a policy is derived for any recall target. 
A single upfront inference produces the fine-grained predictions, departing from recent approaches that invoke models iteratively during search or design separate predictors per recall target. Our formulation predicts normalized recall values, and combined with scale-invariant features and pre-search inference, admits richer architectures and target-agnostic deployment. QASP’s fine-grained progress predictions also enable a lightweight reactive complement that treats each progress estimate as a testable hypothesis and adjusts search depth based on deviations between predicted and observed discovery rates, without additional model inference.

QASP is designed for partitioning-based indices, e.g., Inverted File (IVF), clustering, quantization, multi-dimensional trees, due to their discrete, independent units that enable highly efficient upfront recall estimation with minimal training data, and enable proactive parameter selection and resource allocation. Partitioning-based methods are widely preferred in production systems due to their memory efficiency, parallelization across computational and storage units, compatibility with distributed architectures, and ability to achieve sublinear scaling \cite{sunscaling}.

\begin{figure}[t]
    \centering
    \includegraphics[width=0.9\linewidth]{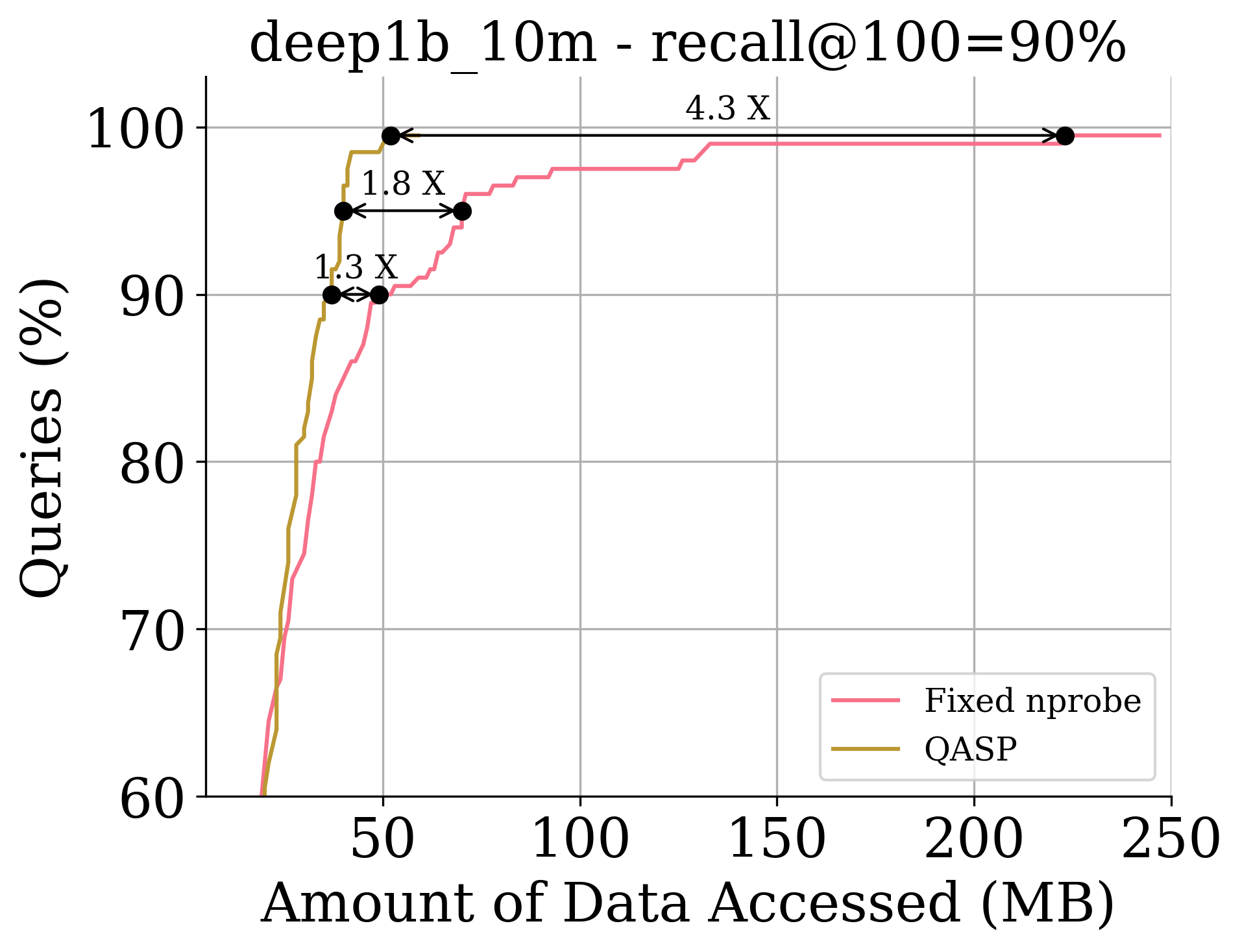}
    \caption{Comparing fixed and adaptive policies. For fixed policy, we gradually increase the number of probes (\texttt{nprobe}) and plot query satisfaction rate (Section~\ref{sec:eval_framework}) versus data accessed on Deep1B 10M dataset. QASP adaptively sets per-query nprobe for the same satisfaction rate.
    Fixed policy consistently over-reads, with the gap widening as more queries reach recall target.}
    \label{fig:motivation}
\vspace{-10pt}
\end{figure}

We establish theoretical guarantees for QASP. We prove that a finite sample of queries suffices for training, independent of dataset size and dimensionality. We show a fundamental lower bound on the population loss of any fixed probe policy, determined by the variance in query difficulty, and prove that QASP matches or improves on the best fixed-policy loss up to a statistical error that vanishes as the training sample grows. This establishes generalization and competitive performance guarantees. Beyond loss, we show that QASP achieves strictly lower expected data access than fixed probing, with the savings growing exponentially in the data's intrinsic  dimensionality as query difficulty heterogeneity increases.

To our knowledge, this is the first work to explore transfer learning and domain adaptation in vector search. Prior solutions are mostly coupled to index configurations and recall targets. We capture scale-invariant query features and predict normalized recall values to support generalization across index configurations, recall targets, and datasets. Through ablation studies and feature importance analysis, we identify the most influential features for recall prediction.  

Experimental evaluation confirms that QASP consistently outperforms baselines and achieves around $56.2\%$ lower recall variance, $33.3\%$ lower deviation from recall target, and $10.9\%$ higher query satisfaction rate with similar or lower data access. QASP transfers effectively across datasets and index configurations with zero-shot or minimal fine-tuning, and scales to larger data and to hierarchical indices achieving 99\% recall with 80\% less data access.

To summarize, this paper makes the following contributions:
\begin{CompactItemize}
\item We introduce QASP, a learning-based query-adaptive search method that predicts fine-grained recall progress, from which search policies are derived for any recall target. It enables proactive optimization beyond early termination and supports cross-domain transfer. The progress predictions also enable a lightweight monitoring mechanism that compares predicted and observed discovery rates during search and adjusts search depth accordingly.

\item 
We establish theoretical bounds on sample size to train QASP to near optimality. We show that fixed policies leave an inherent sub-optimality gap, and prove that QASP is guaranteed to perform within this gap up to a negligible margin. We further prove that QASP achieves strictly lower expected data access, with savings growing exponentially in intrinsic dimensionality.

\item We conduct experiments using both standard and query variability-aware evaluation metrics. QASP achieves significant performance improvements over baselines, with further gains on large datasets, at high recall targets, and for hard queries, where traditional approaches struggle with both consistency and cost.

\end{CompactItemize}

\section{Query-Adaptive Robust Vector Search Policy}
We formalize a learning framework to design a robust search policy based on predicting fine-grained recall contributions of each index partition to optimize per-query performance. Ideal ML-enhanced search should support low inference overhead. Partitioning-based data organization offers a well-suited setting for this by decomposing search into units whose recall contributions can be predicted separately and upfront. We present our approach concretely for clustering/IVF (Inverted File) indices \cite{ferhatosmanoglu2001clusterindex, vqindex2002, jegou2010product, chen2021spann, 
manohar2024parlayann}, though it extends to any partitioning-based organization, including hierarchical clustering as we study in Section~\ref{sec:other_indices}. 
Partitioning-based indices are a core component of large-scale vector databases \cite{chen2021spann, vexless, jayaram2019diskann, oakley2025squash, manohar2024parlayann}.

An IVF-type index $I(D)$ built over a dataset $D \subset \mathbb{R}^d$ is designed to reduce the computational cost of vector similarity search. To this end, it partitions $D$ into $L$ non-overlapping partitions (clusters) $C_1,\ldots,C_L$ with corresponding centroids $\mu_1,\ldots,\mu_L$. Given a query $q\in\mathbb{R}^d$, exhaustive search is performed only on the $l$ top-ranked clusters, where $l \ll L$, and the clusters are ranked by their centroid distance from the query. This raises the question of how to set $l$, which is called the number of \emph{probes}. Commonly, $l$ is set to a fixed value for all queries (see Section~\ref{sec:related}), which we refer to as \emph{fixed probe policy}. Our goal is to leverage supervised learning in order to predict the optimal $l$ per query.

To start, we formally define query-dependent probe policies.
\begin{definition}
    A \emph{probe policy} is a map $g:\mathbb{R}^d\times[0,1]\rightarrow\mathbb{N}$, which maps a query $q\in\mathbb{R}^d$ and a recall target $r\in[0,1]$ to a number $l=g(q,r)$ of top-ranked clusters to probe for $q$.   
    The policy is \emph{valid} for an index $I(D)$ with $L$ clusters if $1\leq g(q,r)\leq L$ for all $q\in\mathbb{R}^d$ and $r>0$ (thus, the policy instructs us to probe at least one cluster and no more than all clusters per query).
\end{definition}

Given a query set $Q\subset\mathbb{R}^d$ and recall target $r^*\in[0,1]$, we say that a policy $g$ is \emph{optimal} for $Q$ and $r^*$ if it probes exactly the minimal number of clusters necessary to achieve recall $r^*$ on each query $q \in Q$. Formally, let $R_{I(D)}(q,l)$ denote the recall attained for a query $q$ when probing its $l$ top-ranked clusters in $I(D)$. Then, the optimal number of probes $\ell^*$ can be defined as follows:
 \begin{equation}\label{eq:opt_probes}
 \ell^*(q,r) := \min\{l:R_{I(D)}(q,l)\geq r\} . 
 \end{equation}
An optimal policy $g^*$ for $Q$ and $r^*$ satisfies $g^*(q,r^*) = \ell^*(q,r^*)$ for all $q\in Q$. 
This does not specify the behavior of $g$ on queries $q\in\mathbb{R}^d\setminus Q$, i.e., it does not specify how to generalize beyond $Q$. 

\subsection{Ordinal Regression}
Our goal is thus to find a policy which is optimal (or close to it) on $Q$, while also generalizing to unseen queries $q\in\mathbb{R}^d\setminus Q$ and to recall values beyond $r^*$. 
This can be naturally cast as a supervised learning problem, where we train a machine learning model to predict the behavior of the index $I(D)$ on the workload $Q$, using its observed performance on $Q$ as well as the ground-truth nearest neighbors of $Q$ in $D$ as labels to train on. 
Supervised learning is useful here for two key properties: on the one hand, ML models can be efficiently trained to find high-quality solutions in the large function spaces, like that of probe policies; on the other hand, if trained with suitable features, ML models are known to generalize well to unseen inputs.

We now formalize an appropriate supervised learning problem. Let $\Theta$ be the parameter space for a parameterized class $\mathcal{G}=\{g_\theta:\theta\in\Theta\}$ of probe policies. To find an optimal probe policy in $\mathcal{G}$ for queries $Q$ and recall target $r^*$, consider the optimization problem:
\begin{equation}\label{eq:learn_g}
  \theta^* = \mathrm{arg}\min_{\theta\in\Theta}\frac1{L^2|Q|}\sum_{q\in Q}\left(\ell^*(q,r^*) - g_{\theta}(q,r^*)\right)^2 . 
\end{equation}
This is an ordinal regression problem, since its goal is to learn a function $g_{\theta^*}$ whose output is an ordinal value $l\in\{1,\ldots,L\}$. The normalization by $L^2|Q|$ ensures that the minimized loss is in $[0,1]$ for convenience and does not alter the minimization problem.  

\subsection{Recall-based Regression}\label{sec:recallbased}
Since ordinal regression is notoriously difficult to optimize \cite{pedregosa2017consistency}, we develop an alternative recall-based approach for learning the probe-policy. To this end, we learn a \emph{recall predictor} as an intermediate step. It maps a query $q$ and number of probes $l$ to the recall it expects to achieve for that query by probing the index its $l$ top-ranked clusters. A recall predictor gives rise to a probe policy by using the minimal number of probes predicted to attain the recall target. 
\begin{definition}
    Given an index $I(D)$ with $L$ clusters, a \emph{recall predictor} is a function $f:\mathbb{R}^d\times\mathbb{N}\rightarrow[0,1]$ that maps a query and probe count to a predicted recall value.
\end{definition}
\begin{definition}
    A recall predictor $f$ induces a probe policy $g^{[f]}$ defined as
    \begin{equation}\label{eq:inducedpolicy}
    g^{[f]}(q,r) := \min\{l:f(q,l)\geq r\}.
    \end{equation}
\end{definition}

Let $\widetilde\Theta$ parameterize a class of recall predictors, $\mathcal{F}=\{f_{\theta}:\theta\in\widetilde\Theta\}$. The induced class of parameterized probe policies is $\mathcal{G}=\{g_{\theta}=g^{[f_{\theta}]}:\theta\in\widetilde\Theta\}$. 
We learn a recall predictor by optimizing an L2 loss:
\begin{equation}\label{eq:learn_f}
  \theta^* = \mathrm{arg}\min_{\theta\in\Theta}\frac1{L|Q|}\sum_{q\in Q}\sum_{l=1}^L\left(f_{\theta}(q,l)-R_{I(D)}(q,l)\right)^2 . 
\end{equation}

Learning a recall predictor as a surrogate for learning a probe policy is justified by the following consistency property.
\begin{proposition}
    Suppose $\theta^*\in\widetilde\Theta$ is such that $f_{\theta^*}$ attains zero loss in eq.~(\ref{eq:learn_f}). Then $g_{\theta^*}=g^{[f_{\theta^*}]}$ attains zero loss in eq.~(\ref{eq:learn_g}).
\end{proposition}
\begin{proof}
    Zero loss in eq.~(\ref{eq:learn_f}) implies that $f_{\theta^*}(q,l)=R_{I(D)}(q,l)$ for all $q,l$. Plugging this in eqs.~(\ref{eq:opt_probes})~and~(\ref{eq:inducedpolicy}) yields $\ell^*(q,r)=g^{[f_{\theta^*}]}(q,r)$ for all $q,r$. This yields zero loss in eq.~(\ref{eq:learn_g}).
\end{proof}

Recall-based regression (eq.~(\ref{eq:learn_f})) offers numerous advantages over direct ordinal regression for the probe policy (eq.~(\ref{eq:learn_g})). First, it is more tractable to optimize \cite{pedregosa2017consistency}. Second, since its target predictions are values in a fixed range $[0,1]$, independent of the number of clusters $L$, the same induced policy is transferrable across indices with different numbers of clusters. Third, the learning problem in eq.~(\ref{eq:learn_f}), unlike that in eq.~(\ref{eq:learn_g}), is independent of a recall target $r^*$, and thus generalizes better across different recall values. Fourth, the loss function in eq.~(\ref{eq:learn_f}) directly involves $L$ numerical values obtained from the index $I(D)$ per query (the values $\{R_{I(D)}(q,l):l=1,\ldots,L\}$), compared to just one such value in eq.~(\ref{eq:learn_g}) (the value $\ell^*(q,r^*)$). This allows the learning model to use more supervised labels per query from the index, increasing ground-truth utilization and learning effectiveness. Directly predicting nprobe couples the model to a specific recall target and index size, and provides $L\times$ fewer supervision signals per query. We confirm empirically (in Section~\ref{sec:baselines}) that single-nprobe prediction results in significantly higher variance and more data access than QASP (Table~\ref{tab:baseline}).
Fifth, recall-based regression inherently produces a recall prediction per cluster, and these values can be leveraged at search time for further optimizations without model inference, such as reactive early/late termination if rate of newly discovered neighbors per probe subceeds/exceeds (respectively) the initial recall predictions (see Section~\ref{sec:reactive_complement}).

\subsection{Model Training}\label{sec:model_training}
QASP employs supervised learning to model the relationship between query characteristics and retrieval performance. Training data generation and model selection are as follows.
\subsubsection{Training Data Generation}
We obtain training signals by searching the index at different probe counts for queries sampled from the dataset, or query workload if available.
We measure recall at every searched partition and use it as the prediction target. Exact $k$-NN search or a high-recall ANN configuration can be used to provide ground-truth labels. We obtain query features as a by-product of this process with no additional pass (see Section \ref{sec:features}). The resulting dataset is split into training and validation sets and used to fit a QASP model.  

The signal-generation process can reuse the data pass performed during index construction. It can also be performed on a smaller index.

\subsubsection{Model Selection}\label{sec:model_variants}
We build three variants of pre-trained QASP models. First, \emph{QASP-DL} is a lightweight deep learning architecture for tabular learning \cite{gorishniy2022embeddings} with batch normalization and dropout. Second, \emph{QASP-GBDT} is a gradient boosted decision tree model which is a de-facto architecture for tabular learning tasks due to their ability to handle decision manifolds effectively \cite{shwartz2022tabular}. Third, \emph{QASP-LITE} is a polynomial regression model. All variants are trained to optimize mean squared error loss and differ in expressive power, inference latency, ease of integration and adaptability to unseen domains.

\subsection{Model Invocation}\label{sec:model_inference}
Training QASP models on the recall regression task serves as a surrogate for the search policy during model invocation Section~\ref{sec:recallbased}. Beyond the naive application of QASP to the same data distribution as it is trained on, QASP can be effectively employed at search-time.

\subsubsection{Model Inference}\label{sec:batch_inference}
QASP model inference is performed proactively on the query workload prior to vector search execution. This design allows rich architecture such as deep learning unlike prior work \cite{chatzakis2025darth} due to workload-level optimizations. Batching for large-scale inference allows easy parallelism and GPU acceleration can yield significant scaling benefits. By processing queries in batches, the approach amortizes computational overhead and sustains high throughput, even as the index partitioning scales to orders of magnitude beyond the size observed during training.

\subsubsection{Domain Adaptation}\label{sec:generalization}
Recall prediction serves as a foundational training objective for vector search, generalizing effectively across datasets, query distributions, and index configurations.  We design features that are normalized and scale-invariant (see Section \ref{sec:features}), which capture fundamental search difficulty signals that transfer across heterogeneous scales and indexing schemes. This enables QASP to adapt to new domains via both zero-shot transfer, where the pre-trained model generalizes to unseen datasets and index configurations without adjustment, and few-shot transfer, with limited training on batch normalization \cite{li2016revisitingbatchnormalizationpractical} and final projection layers using very few queries from the target domain.

\subsubsection{Scaling to Hierarchical Indices}\label{sec:other_indices}
At invocation time, QASP can be scaled to large indices without additional training by extending the same feature transformations on hierarchical partitioning. QASP models can effectively be trained once on a smaller flat index and applied across all hierarchy levels. At inference, model features at each level are computed using coresets, with second-level centroids already available from index construction providing a lightweight approximation. This combination enables fast, scalable model invocation across large indices.

\subsection{Lightweight Reactive Complement}
\label{sec:reactive_complement}
QASP's fine-grained recall predictions $R_1,\ldots,R_L$ enable real-time validation as search progresses: Each prediction $R_i$ represents cumulative recall after searching $i$ partitions, making the incremental contribution $\widetilde{DR}_i = R_i - R_{i-1}$ a testable hypothesis about partition $i$'s value. For example, when QASP predicts that partition $i$ will increase recall from 0.89 to 0.94, it asserts that this partition contains approximately $0.05 k$ true nearest neighbors, which can be immediately verified during search. The actual discovery rate $DR_i$ is observed as partition $i$'s contribution to the running result set.

The goal is to detect whether predicted rates $\widetilde{DR}_i$ under-predict or over-predict observed rates $DR_i$ as search progresses. We adapt a lightweight statistical process control solution using EWMA smoothing: $SDR_i = \alpha \cdot DR_i + (1-\alpha) \cdot SDR_{i-1}$ to prioritize recent discovery rates. A Shewhart-type run rule detects consistent deviations $\Delta_i = (SDR_i - \widetilde{DR}_i) /\widetilde{DR}_i$ over consecutive steps. Consistently negative values indicate early termination opportunities, while consistently positive values suggest extending search beyond initial predictions. See Section \ref{sec:experiments_reactive} for pseudocode.

\section{Model Features}\label{sec:features}

Model features are data signals that are informative for recall prediction task and accessible at inference with no additional computation. We define a set of features, with scale-invariant transformations, designed to capture diverse query and index characteristics.

\subsection{Feature Types}\label{sec:feature_types}
\emph{Dataset Features.}
\noindent \textit{Data Size}: ${1} / {\sqrt[d]{n}}$, where $n=|D|$, is the dataset size and $d$ is the vector dimensionality. This quantity captures the expected separation (min distance) between $n$ points packed in a unit sphere.

\emph{Index Features.}
\noindent \emph{Rank}: $1 - \exp(-l/L)$, where $l$ is centroid rank and $L$ is total number of centroids. Normalization by $L$ ensures rank $\in (0, 1-1/e]$ for all indices. Exponential decay captures diminishing returns of accessing further clusters.

\noindent \emph{Cumulative Cluster Size}: ${\sum_{i=1}^{l} |C_{\pi_q(i)}|} / {|D|}$. This represents fraction of vectors read up to a given centroid rank $l$. This feature is especially useful for imbalanced indices.

\noindent \emph{Cluster Coefficient of Variance}: $\sigma(d(v_l,\mu_l)) / \mu(v_l,\mu_l)$, where \\
$\sigma(d(v_l,\mu_l))$ is the standard deviation and $\mu(d(v_l,\mu_l))$ is the mean distance between data vectors $v_l \in C_l$ and centroid $\mu_l$.

\noindent \emph{Maximum Cluster Distance}: $\max(v_l,\mu_l) / d(q, \mu_{\pi_q(1)})$, where\\
$\max(v_l,\mu_l)$ is the max distance between data vectors $v_l \in C_l$ and centroid $\mu_l$ and $d(q, \mu_{\pi_q(1)})$ is the distance between query and the centroid of nearest cluster. We also consider other statistical summaries (like $min, p25, p75$) over the distance distribution. Normalization with $(q, \mu_{\pi_q(1)})$ is performed at search time.

\emph{Query Features.}
\noindent \emph{Relative Distance}: $\tanh(\delta(q,l)-1)$, where $\delta(q,l) = d(q,\mu_l)/d(q, \mu_{\pi_q(1)})$ and $d(q, \mu_{\pi_q(1)})$ is the distance from query $q$ to its nearest centroid. Query-to-centroid distances are normalized relative to the nearest centroid distance. 

Hyperbolic tangent compresses distances into bounded ranges and aims to capture diminishing returns of accessing distant clusters.

\noindent \emph{Local Relative Contrast}: ${\mu(d(q, \mu_{\pi_q(\cdot)}))} / {d(q, \mu_{\pi_q(1)})}$,  where \\
$d(q, \mu_{\pi_q(1)})$ is the distance from query $q$ to its nearest centroid and the numerator represents the average distance to all other centroids. This feature is an adaptation of the local relative contrast defined by \cite{aumuller2021role} using centroid vectors. We also adapt \emph{Local Intrinsic Dimensionality} similarly using the formula $-L / \sum_{i=1}^{L} [\ln d(q,\mu_{\pi_q(i)}) - \ln d(q,\mu_{\pi_q(L)})]$, where $d(q,\mu_{\pi_q(i)})$ is the distance from query $q$ to the $i$th centroid and $d(q,\mu_{\pi_q(L)})$ is the distance between $q$ and the farthest centroid. However, this feature is ablated from the final feature set (see Section \ref{sec:feature_studies}).

\subsection{Feature Transformation}\label{sec:feature_transformation}

The foregoing feature design incorporates various transformations (e.g., normalization with respect to $L$ or $|D|$, using exponential decay and hyperbolic tangent) to bound feature ranges within fixed intervals, thereby providing a scale-invariant feature representation across varying search configurations. By expressing features as normalized ratios, they remain meaningful across different vector space distributions, dataset scales or index configurations. Another transformation we consider is \emph{relative difference} or ``relative jumps'' for any feature, call it $\varphi(q,l)$, varying with $l$:
\[
  \Delta(\varphi;q,l) = \frac{\varphi(q,l) - \varphi(q,l-1)}{\max_{1 < j < l}(\varphi(q,j) - \varphi(q,j-1))}.
\]
This allows QASP to rely not only on the raw feature value but also on how much the feature value changes as more clusters are probed in order to detect diminishing returns. We pad the output of this transformation with $0$ wherever undefined. We conduct feature analysis and ablations to find the  feature set (see Section \ref{sec:feature_studies}).
\section{Theoretical Guarantees for QASP}\label{sec:learning_theory}
In this section we analyze QASP's optimization and connect the structure of the vector index to the provable advantage of QASP. We first show that every fixed policy suffers irreducible loss determined by the variance of query difficulty, and that QASP can be trained nearly optimally with a finite number of samples independent of dataset size and never outperformed by any fixed policy except by a vanishing margin. We then derive a dominance condition proving that QASP's data access savings grow exponentially in intrinsic dimensionality. 

Let $\mathscr Q$ be a query distribution over $\mathbb{R}^d$. Suppose that the given query set $Q$ is a sample from $\mathscr Q$. Let $r^*\in[0,1]$ be the recall target. We define the population loss of a probe policy $g$ on $\mathscr Q$ and $r^*$ as
\begin{equation}\label{eq:populationloss}
  \mathcal L(g) := \mathbb{E}_{q\sim\mathscr Q}\left[\frac{1}{L^2}(\ell^*(q,r^*) - g(q,r^*))^2\right] ,
\end{equation}
where $\ell^*(q,r^*)$ is the minimum number of clusters in $I(D)$ that need to be probed to attain recall $r^*$ for $q$, as per eq.~(\ref{eq:opt_probes}). For a finite sample $Q$ from $\mathscr Q$, we define the corresponding empirical loss $\hat{\mathcal L}(g)$ as follows. Observe that the probe policy learning problem in eq.~(\ref{eq:learn_g}) minimizes this empirical loss on $Q$.
\[
  \hat{\mathcal L}(g) := \frac1{L^2|Q|}\sum_{q\in Q}(\ell^*(q,r^*) - g(q,r^*))^2 .
\]

\subsection{Sub-Optimality of Fixed Probe Policies}\label{sec:fixed_suboptimal}
We first provide a suboptimality bound for the population loss of fixed policies. A fixed policy $g$ maps all queries to a fixed number of probes. We denote by $g^{(l)}$ the fixed $l$-probe policy, defined as
\begin{equation}\label{eq:fixedpolicy}
  \forall q\in\mathbb{R}^d, \quad g^{(l)}(q,r^*) = l .
\end{equation}

We define a distribution over $\{1,\ldots,L\}$ from $\mathscr Q$. Let $\Lambda$ be a random variable in $\{1,\ldots,L\}$ drawn by sampling $q\sim\mathscr Q$ and setting
\begin{equation}\label{eq:lambdadef}
    \Lambda=\ell^*(q,r^*) .
\end{equation}

Thus, $\Lambda$ is the optimal number of clusters that needs to be scanned for a query drawn at random from $\mathscr Q$, under the recall target $r^*$. We can relate the sub-optimality of fixed probe policies to the distribution of $\Lambda$ as follows. 

\begin{theorem}\label{prp:lambdabound}
For every $l\in\{1,\ldots,L\}$, the population loss of the fixed-probe policy $g^{(l)}$ satisfies $\mathcal{L}(g^{(l)}) \geq \mathrm{Var}(\Lambda)/L^2$.
\end{theorem}

Thus, the single quantity $\frac{\mathrm{Var}(\Lambda)}{L^2}$ is a uniform lower bound on the population losses of every fixed probe policy simultaneously. 
We can evaluate this quantity --- either empirically from the given query workload, or analytically under a distributional query model --- to quantify the sub-optimality enforced by restricting the search procedure to fixed-probe policies.

\paragraph*{Example---Zipfian probes}
The $\rho$-Zipfian law, $\Pr[\ell]\propto 1/\ell^\rho$, is commonly used to model access frequency patterns \cite{rivest1976self}. For simplicity, we focus on the most standard case $\rho=1$. Thus, for a partition-based vector index with $L$ clusters and a query drawn at random from $\mathscr Q$, the probability that its optimal number of probes is $\ell$ is proportional to $1/\ell$. Hence, $\Lambda$ is distributed as $\Pr[\Lambda=\ell] = 1/(\ell\cdot H_L)$, where $H_L=\sum_{i=1}^L\frac1i=\Theta(\ln L)$ is the $L^{th}$ harmonic number. 

A direct calculation yields that
\[ 
  \mathrm{Var}(\Lambda) = 
  \frac{L(L+1)}{2H_L} - \frac{L^2}{H_L^2} = \Theta\left(\frac{L^2}{H_L}\right) = 
  \Theta\left(\frac{L^2}{\ln L}\right).
\]
Hence, by Theorem~\ref{prp:lambdabound}, the population loss (\ref{eq:populationloss}) of every fixed-probe policy is at least $\Omega(1/\ln L)>0$.

\begin{proof}[Proof of Theorem \ref{prp:lambdabound}]
We recall the following basic fact:
\begin{lemma}\label{prp:minmean}
    For every real-valued random variable $X$ with $\mathbb{E}[X]<\infty$, it holds that $\arg\min_{\mu\in\mathbb{R}}\mathbb{E}\left[(\mu-X)^2\right]=\mathbb{E}[X]$. 
\end{lemma}
Using this, every fixed probe policy $g^{(l)}$ satisfies,
\begin{align*}
\mathcal{L}(g^{(l)}) &= \mathbb{E}_{q\sim\mathscr Q}\left[\frac{1}{L^2}(g^{(l)}(q,r^*) - \ell^*(q,r^*))^2\right] & \text{eq.~(\ref{eq:populationloss})}\\
&= \frac{1}{L^2}\mathbb{E}_{q\sim\mathscr Q}\left[(l - \ell^*(q,r^*))^2\right] & \text{eq.~(\ref{eq:fixedpolicy})}\\
&\geq \frac{1}{L^2}\min_{\mu\in\mathbb{R}}\mathbb{E}_{q\sim\mathscr Q}\left[(\mu - \ell^*(q,r^*))^2\right] & \\
&= \frac{1}{L^2}\mathbb{E}_{q\sim\mathscr Q}\left[\left(\mathbb{E}_{q\sim\mathscr Q}[\ell^*(q,r^*)] - \ell^*(q,r^*)\right)^2\right]\\
&= \frac{1}{L^2}\mathbb{E}\left[\left(\mathbb{E}\left[\Lambda\right] - \Lambda\right)^2\right] = \frac{\mathrm{Var}(\Lambda)}{L^2}.  & \text{eq.~(\ref{eq:lambdadef})
}
\qedhere
\popQED
\end{align*}
\end{proof}

\subsection{QASP versus Fixed Probe Policies}\label{sec:qasppdim}
We show that a finite sample of queries $Q$ from $\mathscr Q$ suffices to train QASP to better loss than any fixed probe policy with high probability, up to a small gap that vanishes as the sample size grows. Importantly, the requisite sample size is independent of the dataset size $|D|$, and depends only on the number of clusters $L$ and on the structure of the QASP model.
We start with some learning-theoretic definitions. The pseudo-dimension characterizes the number of samples required to certify generalization \cite{pollard1984convergence,pollard1990empirical,blumer1989learnability,gupta2017pac,balcan2020data}. We define pseudo-dimension for probe policies as follows:

\begin{definition}[pseudo-dimension of probe policy class]
    Let $\mathcal G$ be a class of probe policies. 
    Let $Q=\{q_1,\ldots,q_M\}\subset\mathbb{R}^d$. 
    We say that $Q$ is \emph{pseudo-shattered} by $\mathcal G$ if there exist thresholds $\tau_1,\ldots,\tau_M\in\mathbb{R}$ such that for every $I\subset[M]$, there is $g\in\mathcal G$ such that $g(q_i,r^*)>\tau_i\Leftrightarrow i\in I$. 
    The \emph{pseudo-dimension} of $\mathcal G$ is denoted $\mathrm{pdim}(\mathcal G)$ and defined as the largest size of a set $Q\subset\mathbb{R}^d$ pseudo-shattered by $\mathcal G$.
\end{definition}

Function classes with finite pseudo-dimension are called \emph{learnable} because their population loss can be approximately optimized from a finite sample. 
Specifically, the pseudo-dimension governs the sample size required to guarantee that the empirical loss approximates the population loss for every probe policy in the class, a property known as \emph{uniform convergence} \cite{vapnik2015uniform}.
\begin{theorem}[\cite{anthony2009neural}, Theorem 19.2]\label{thm:uniform_convergence}
    Let $\varepsilon,\eta\in(0,1)$. Let $Q$ be a sample from $\mathscr Q$ of size $O((\mathrm{pdim}(\mathcal G)+\log(1/\eta))/\varepsilon^2)$. Then,
    \[
      \Pr\left[\forall g\in\mathcal G:\;|\mathcal L(g) - \hat{\mathcal L}(g)|<\varepsilon\right] \geq 1-\eta .
    \]
\end{theorem}

The pseudo-dimension can be bounded using the dimension
of the parameter space and the complexity of the
learned model, as per the following theorem, which is a specialization of Theorem 8 from \cite{bartlett2003vapnik} (based on  \cite{goldberg1995bounding,karpinski1997polynomial}) to probe policies.

\begin{theorem}\label{thm:goldbergjerrum}
Let $\mathcal G=\{g_\theta:\theta\in\Theta\}$ be a class of probe policies parameterized by an $m$-dimensional parameter space $\Theta$. Let $F(q)$ be any set of numerical features representing a query $q\in\mathbb{R}^d$. 
Suppose that for every $\theta\in\Theta$ and $q\in\mathbb{R}^d$, it is possible to compute $g_\theta(q,r^*)$ from $\theta$ and $F(q)$ with an algorithm that can perform arithmetic operations $\{+,-,\times,\div\}$ and numerical comparisons $\{=,\neq,>,\geq,<,\leq\}$. Suppose the algorithm has running time $t$. Then, $\mathrm{pdim}(\mathcal G)=O(mt)$. 
Furthermore, if the algorithm is also allowed to perform exponentiation ($x\mapsto e^x$), then $\mathrm{pdim}(\mathcal G)=O((mt)^2)$. 
\end{theorem}

We show that our implicit ML-based search policy, where probe decisions are induced through recall predictions can be analyzed through bounded pseudo-dimension theory. We prove that QASP's compositionally-defined policy class has finite pseudo-dimension despite complex interactions between neural network recall predictors, query features, and the minimization operation that determines optimal probe counts.

\begin{lemma}\label{prp:qasp_pdim}
    The pseudo-dimension $\mathrm{pdim}(\mathcal G_{\mathrm{QASP}})$ of all three QASP variants is finite. Furthermore, it depends only on $L$, and is independent of the dataset size $n$ and the dimensionality $d$. 
\end{lemma}
\begin{proof}
    Invoke Theorem~\ref{thm:goldbergjerrum}. 
    By Section~\ref{sec:features}, the number of features $|F(q)|$ per query $q$ is $O(L)$, hence so is the number of trainable parameters (i.e., the dimension of the parameter space) per QASP model. 
    By Section~\ref{sec:model_variants}, all three QASP architectures can be computed from $F(q)$ using only arithmetic operations and numerical comparisons, and (in the case of QASP-DL) exponentiations (in order to compute sigmoid activations). Therefore, the proposition follows from Theorem~\ref{thm:goldbergjerrum}.
\end{proof}

\begin{theorem}\label{cor:qasp_vs_fixed}
    Let $\varepsilon,\eta\in(0,1)$. Suppose the sample $Q$ has size $O(\mathrm{pdim}(\mathcal G_{\mathrm{QASP}})/\varepsilon^2+\log(1/\eta))$. Let $g_{\mathrm{QASP}}$ be the policy that minimizes eq.~(\ref{eq:learn_g}). 
    Then, with probability $1-\eta$, it holds that
    \begin{equation}\label{eq:qasp_vs_fixed} 
    \forall\;l\in[L], \quad \mathcal{L}(g_{\mathrm{QASP}}) \leq \mathcal{L}(g^{(l)}) + 2\varepsilon .
    \end{equation}
    Note that both the failure probability $\eta$ and the performance gap $2\varepsilon$ vanish as we increase the sample size $|Q|$.
\end{theorem}

\begin{proof}
With Lemma \ref{prp:qasp_pdim}, we can prove Theorem \ref{cor:qasp_vs_fixed}. Since $g_{\mathrm{QASP}}$ minimizes the empirical loss $\hat{\mathcal{L}}$, we have,\begin{equation}\label{eq:fixedpolicy_empiricalloss}
\forall\;l\in[L], \quad \hat{\mathcal{L}}(g_{\mathrm{QASP}}) \leq \hat{\mathcal{L}}(g^{(l)}) . 
\end{equation} \
By definition, the learned policy $g_{\mathrm{QASP}}$ is realizable in the QASP model, meaning $g_{\mathrm{QASP}}\in \mathcal G_{\mathrm{QASP}}$. 
Furthermore, all fixed-probe policies $g^{(l)}$ are realizable in the QASP model.
Therefore, by Lemma~\ref{prp:qasp_pdim} and Theorem~\ref{thm:uniform_convergence}, we have with probability $1-\eta$,
\[ |\mathcal L(g_{\mathrm{QASP}}) - \hat{\mathcal L}(g_{\mathrm{QASP}})|<\varepsilon \;\; \text{and} \;\; \forall\;l, |\mathcal L(g^{(l)}) - \hat{\mathcal L}(g^{(l)})|<\varepsilon . \]
Together with eq.~(\ref{eq:fixedpolicy_empiricalloss}), with probability $1-\eta$, eq.~(\ref{eq:qasp_vs_fixed}) follows.
\end{proof}

\noindent\emph{Remark.} 
The same conclusion as Theorem \ref{cor:qasp_vs_fixed} holds if we define the empirical loss with eq.~(\ref{eq:learn_f}) instead of eq.~(\ref{eq:learn_g}) and the appropriate analogous population loss (replacing $\Sigma_{q\in Q}$ with $\mathbb{E}_{q\sim\mathscr Q}$ in eq.~(\ref{eq:learn_f})). 
The formal proof goes by augmenting the feature set $F(q)$ with the additional query features $\{R_{I(D)}(q,l)\}_{l=1}^L$, which enables computing the contribution of $q$ to the loss from $F(q)$ while maintaining the size bound $|F(q)|=O(L)$. 
We omit further details.

In Section \ref{sec:fixed_suboptimal} we establish that fixed policies inherently leave room for loss improvement.

Lemma \ref{prp:qasp_pdim} allows us to complement this result and prove that even in the worst case, with high probability, QASP's loss is never \emph{worse} than any fixed policy, except by a negligible margin that tends to zero as the sample size grows.

\subsection{Cost Dominance of Adaptive Probing}

Section~\ref{sec:fixed_suboptimal} and Section~\ref{sec:qasppdim} establish that QASP generalizes from finite samples and that fixed policies are inherently sub-optimal in loss. We now ask a complementary question: \emph{how much data access does QASP save}, and under what conditions does it dominate fixed probing despite predictor imperfection? Fixed policies must provision for the tail of the query difficulty distribution, forcing every query, including easy ones, to pay the cost of the hardest queries. QASP instead pays what each query individually requires. We formalize this intuition and derive the precise conditions under which QASP dominates, connecting the geometry of the vector index to the quality of the learned predictor, with the gap growing exponentially in the intrinsic dimensionality of the data.

We introduce three quantities characterizing the index. Let $D_I = \max_j \max_{x \in C_j} \|x - \mu_j\|$ be the maximum cluster radius, $\sigma_I = \min_{j \neq j'} \|\mu_j - \mu_{j'}\|$ the minimum inter-centroid separation, and $\Delta$ the doubling dimension \footnote{The smallest integer such that every ball of radius $r$ is covered by $2^\Delta$ balls of radius $r/2$.} of the dataset union centroids. Write $r(q) = \|q - x^*(q)\|$ for the nearest-neighbor distance of query $q$, $F_r$ for its CDF over $\mathcal{Q}$, and $r_\delta = F_r^{-1}(1-\delta)$ for the $(1-\delta)$-quantile.

\begin{lemma}[Cell Intersection Bound]\label{lem:cell_bound}
$\ell^*(q, 1) \leq \bigl({4(r(q) + D_I)}/{\sigma_I}\bigr)^{\Delta}$.
\end{lemma}

\begin{proof}
Let $j^*$ be the cluster containing $x^*(q)$. By the triangle inequality, $\|q - \mu_{j^*}\| \leq r(q) + D_I$, so $\mu_{j^*} \in B(q, r(q) + D_I)$. Every centroid ranked before $\mu_{j^*}$ also lies in this ball, so $\ell^*(q,1)$ is at most the number of centroids in $B(q, r(q) + D_I)$. Since centroids are pairwise separated by $\geq \sigma_I$, the balls $\{B(\mu_j, \sigma_I/2)\}$ are disjoint and contained in $B(q, r(q) + D_I + \sigma_I/2)$. By the doubling-dimension packing property~\cite{cover_trees}, a ball of radius $R$ contains at most $(2R/\rho)^\Delta$ points with pairwise distance $\geq \rho$. Setting $R = r(q) + D_I + \sigma_I/2$ and $\rho = \sigma_I$ gives $\ell^*(q,1) \leq (2(r(q) + D_I)/\sigma_I + 1)^\Delta \leq (4(r(q) + D_I)/\sigma_I)^\Delta$ when $r(q) + D_I \geq \sigma_I/2$.
\end{proof}

The ratio $(r(q) + D_I)/\sigma_I$ captures geometric query difficulty: queries far from their nearest neighbor or in poorly separated partitions require more probes. This bound is tight up to constants: by definition of doubling dimension, there exist point configurations where $\Theta((R/\rho)^\Delta)$ points with pairwise separation $\geq \rho$ fit in a ball of radius $R$~\cite{cover_trees}, so the upper bound is achievable. Note that $\ell^*(q, 1)$ bounds the probes needed to retrieve the single nearest neighbor; for $k$-NN recall at target $r^*$, the bound is conservative since achieving partial recall is strictly easier. 

\begin{proposition}[Fixed Probe Requirement]\label{prop:fixed_cost}
A fixed-$\ell$ policy achieves recall $\geq 1 - \delta$ when
\begin{equation}\label{eq:fixed_cost}
\ell \;\geq\; \ell_{\mathrm{fix}}(\delta) \;:=\; \left(\frac{4(r_\delta + D_I)}{\sigma_I}\right)^{\!\Delta}.
\end{equation}
\end{proposition}

\begin{proof}
By Lemma~\ref{lem:cell_bound}, $r(q) \leq \frac{\sigma_I}{4}\,\ell^{1/\Delta} - D_I$ implies $\ell^*(q,1) \leq \ell$. Setting this threshold $\geq r_\delta$ gives~\eqref{eq:fixed_cost}.
\end{proof}

This is the cost every query pays under fixed probing --- including easy queries with $r(q) \ll r_\delta$ that would succeed with far fewer probes. We now bound what QASP pays instead.
Let $\hat{\ell}(q) = g_{\mathrm{QASP}}(q, r^*)$ be QASP's predicted probe count. Define the \emph{failure rate} $\varepsilon = \Pr[\hat{\ell}(q) < \ell^*(q, r^*)]$ and the \emph{overshoot} $\bar{\eta} = \mathbb{E}[\hat{\ell}(q) - \ell^*(q,r^*) \mid \hat{\ell}(q) \geq \ell^*(q,r^*)]$. The failure rate $\varepsilon$ is the probability QASP under-probes (missing the recall target); the overshoot $\bar{\eta}$ is the average wasted probes when it succeeds.

\begin{proposition}[QASP Expected Probes]\label{prop:qasp_cost}
$\mathbb{E}[\hat{\ell}(q)] \leq \mathbb{E}[\Lambda] + \bar{\eta} + \varepsilon \cdot L$.
\end{proposition}

\begin{proof}
Condition on success ($\hat{\ell} \geq \ell^*$) and failure ($\hat{\ell} < \ell^*$):
$\mathbb{E}[\hat{\ell}] = (1-\varepsilon)(\mathbb{E}[\ell^* \mid \text{success}] + \bar{\eta}) + \varepsilon\,\mathbb{E}[\hat{\ell} \mid \text{failure}]$.
Since $(1-\varepsilon)\mathbb{E}[\ell^* \mid \text{success}] \leq \mathbb{E}[\Lambda]$, $(1-\varepsilon)\bar{\eta} \leq \bar{\eta}$, and $\mathbb{E}[\hat{\ell} \mid \text{failure}] \leq L$ (a worst-case bound; in practice failure queries cost far less than $L$ probes), the bound follows.
\end{proof}

\begin{corollary}[Dominance Condition]\label{cor:dominance}
At matched recall ($\varepsilon = \delta$), QASP has strictly lower expected data access than fixed probing whenever
\begin{equation}\label{eq:dominance}
\underbrace{\bar{\eta} + \varepsilon \cdot L}_{\text{predictor overhead}} \;<\; \underbrace{\left(\frac{4(r_\delta + D_I)}{\sigma_I}\right)^{\!\Delta} - \mathbb{E}[\Lambda]}_{\text{tail waste of fixed probing}}.
\end{equation}
\end{corollary}

\begin{figure}[t]
\centering
\includegraphics[width=0.9\columnwidth]{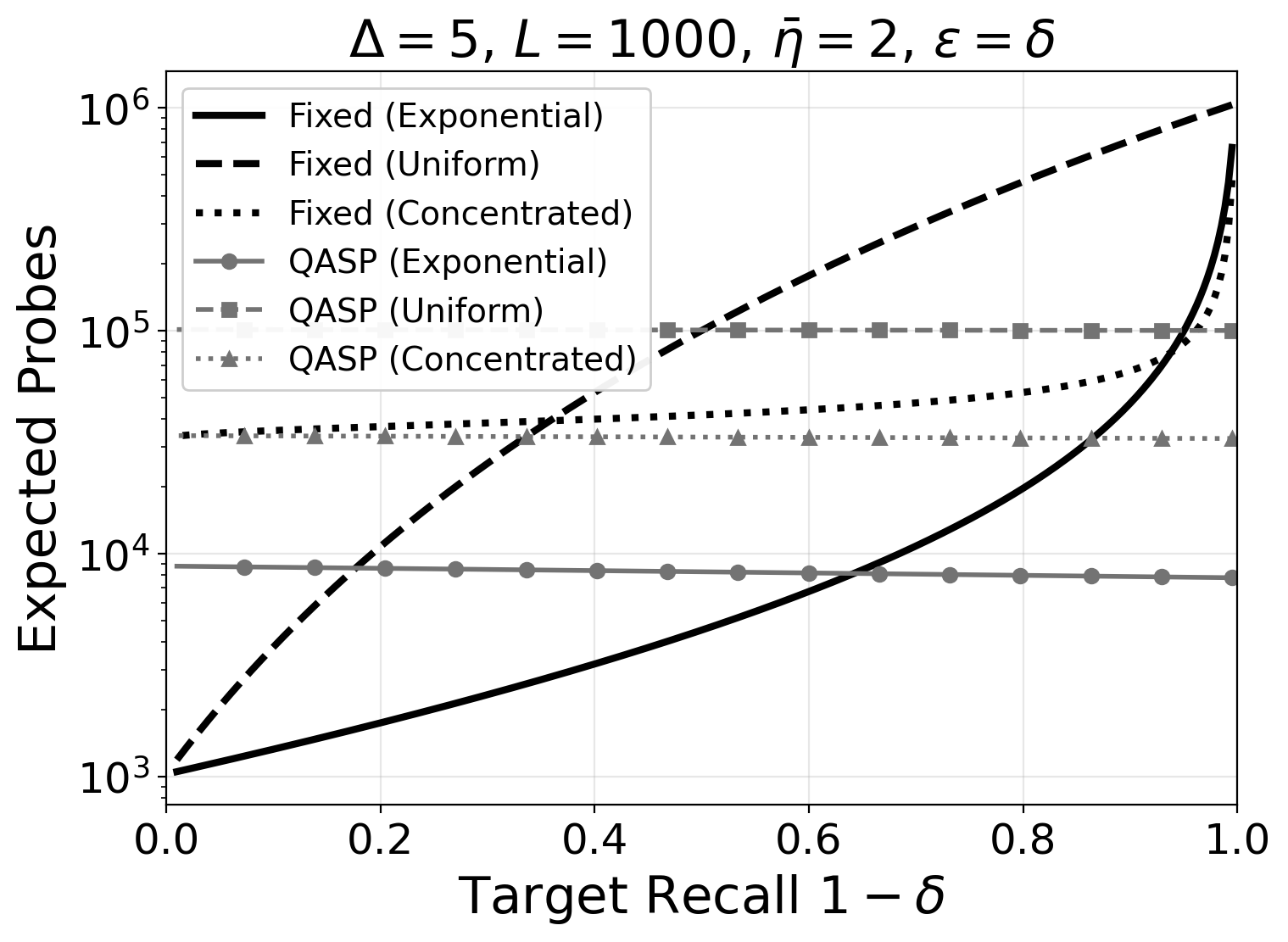}
\caption{Expected probes vs.\ recall target for fixed probing (black) and QASP (gray) under three NN distance distributions ($\Delta=5$, $L=1000$, $\bar{\eta}=2$, $\varepsilon = \delta$). The vertical gap between matched curves is the data access savings from adaptive probing. QASP benefits most when $F_r$ has a heavy tail (exponential, uniform) and at high recall targets, where fixed probing must provision for the worst-case quantile $r_\delta$.}
\label{fig:scan_cost}
\end{figure}

Combining Propositions~\ref{prop:fixed_cost} and~\ref{prop:qasp_cost} yields the main result.
The left side is the total cost of predictor imperfection: wasted probes from over-prediction ($\bar{\eta}$) plus the penalty from under-prediction failures ($\varepsilon L$). The right side is the \emph{tail
waste}: the gap between the fixed probe count (driven by the hardest queries via $r_\delta$) and the average difficulty $\mathbb{E}[\Lambda]$. This gap grows as the $\Delta$-th power of the ratio between worst-case and average query difficulty. 
QASP dominates whenever the predictor is ``good enough'' relative to this
heterogeneity; conversely, when $F_r$ is concentrated, all queries are equally hard and adaptive probing offers no advantage. Figure~\ref{fig:scan_cost} illustrates the dominance condition under three NN
distance distributions: the gap between fixed and QASP curves at any recall target represents data access savings, largest for heavy-tailed distributions at high recall and vanishing for concentrated distributions which is a theoretical worst case scenario. Note that we do not model $\varepsilon L$ and $\bar{\eta}$ as functions of the recall target, which in practice influences predictor performance.

\section{Related Work}\label{sec:related}

QASP targets partitioning-based indexing, the most comprehensive class of vector search methods that partition the space for pruning and scale to distributed architectures. These include clustering and IVF (Inverted File Index) \cite{ferhatosmanoglu2001clusterindex,vqindex2002,jegou2010product,chen2021spann,quake}, multi-dimensional space trees \cite{bentley1975multidimensional,omohundro1989five,zezula1998approximate,cover_trees}, and vector quantization \cite{ferhatosmanoglu2000vector,vqindex2002,jegou2010product} which uses clustering at its core. Partitioning-based indices achieve sublinear scaling \cite{sunscaling} and are well suited to large vector databases because of their parallelizability and resource efficiency. Partitioning also underpins distributed search architectures to distribute vectors across storage units or processors \cite{vexless, jayaram2019diskann, oakley2025squash, manohar2024parlayann}.

Vector search libraries predominantly relied on static empirical parameter settings and target average recall without adapting to query difficulty. These include practices such as defaulting to nprobe = 8 (number of clusters to search), recommending proportional scaling with nlist (total number of clusters) and experimental tuning \cite{milvus_guide}; setting nprobe to cover 5--10\% of the dataset \cite{lancedb_guide}; setting nlist based on dataset size ({$4 \cdot \sqrt{N}$ \text{ to } $16 \cdot \sqrt{N}$}) with limited nprobe guidance \cite{faiss_guidelines} and testing multiple values (1, 4, 16, 64, 256) \cite{faiss_autotune_example}. SPANN \cite{chen2021spann} and SQUASH \cite{oakley2025squash} use thresholds to decide which clusters to search.

Popular graph-based methods \cite{malkov2018efficient} constrain the candidate list size at query time, traditionally fixed for all queries. Recent methods \cite{laet,chatzakis2025darth,adaef2026} set this parameter per query through either statistical scoring or reactive early termination. Statistical scoring offers lightweight runtime adaptation but requires distributional assumptions that may not hold across diverse workloads. Reactive early termination provides flexible stopping criteria but periodically invokes models during traversal and incurs per-step inference overhead that constrains practical model size and complexity. These methods target graph indices.  For IVF indices, PCE-Net~\cite{PceNet} predicts a single nprobe per query, which couples the prediction to a specific recall target and results in higher variance and data access (Section~\ref{sec:baselines}). QASP instead predicts fine-grained recall contributions over partitions proactively with a single upfront inference. This enables the use of richer model architectures and supports generalization across recall targets.

\begin{table*}[t]
\centering
\caption{Comparison of QASP-DL and Oracle Nprobe policy across three recall targets.
Oracle Nprobe achieves recall target by over-reading in the 50\% Easy split and under-reading in the 50\% Hard split, while QASP optimizes search on both splits.
We highlight $\bar{r}$ with large deviation from $r^*$ in red otherwise in teal, while best average QVE metrics are emboldened when $\bar{r} \approx r^*$. These results empirically validate Corollary \ref{cor:dominance} where $A\%$ improvements are better at higher recall targets.}\label{table:main}
\resizebox{\textwidth}{!}{%
\begin{tabular}{ll|ccccc|ccccc|ccccc}
\toprule
\multirow{2}{*}{\textbf{Dataset}} & \multirow{2}{*}{\textbf{Search Policy}}
& \multicolumn{5}{c}{\(\mathbf{r^* = 90\%}\)} 
& \multicolumn{5}{c}{\(\mathbf{r^* = 95\%}\)} 
& \multicolumn{5}{c}{\(\mathbf{r^* = 99\%}\)} \\
\cmidrule(lr){3-7} \cmidrule(lr){8-12} \cmidrule(lr){13-17}
& & \(\bar{r}\) & \(\sigma^2\downarrow\) & \(\delta\downarrow\) & \(S\%\uparrow\) & \(A\%\downarrow\)
            & \(\bar{r}\) & \(\sigma^2\downarrow\) & \(\delta\downarrow\) & \(S\%\uparrow\) & \(A\%\downarrow\)
            & \(\bar{r}\) & \(\sigma^2\downarrow\) & \(\delta\downarrow\) & \(S\%\uparrow\) & \(A\%\downarrow\) \\
\midrule
\multirow{1}{*}{{AVERAGE}} & Oracle Nprobe & \textcolor{teal}{89.84} & 141.39 & 9.45 & 72.19 & \textbf{3.88} & \textcolor{teal}{95.04} & 49.58 & 5.25 & 80.80 & 7.06 & \textcolor{teal}{98.43} & 9.93 & 1.67 & 92.18 & 14.81 \\ 
 & QASP-DL & \textcolor{teal}{90.87} & \textbf{61.95} & \textbf{6.30} & \textbf{80.03} & 3.96 & \textcolor{teal}{95.09} & \textbf{28.45} & \textbf{3.87} & \textbf{85.60} & \textbf{5.96} & \textcolor{teal}{98.46} & \textbf{7.18} & \textbf{1.36} & \textbf{94.56} & \textbf{11.47} \\
\multirow{1}{*}{50\% EASY}  & Oracle Nprobe & \textcolor{red}{96.47} & 28.51 & 8.02 & 93.41 & 4.08 & \textcolor{red}{98.57} & 7.45 & 4.26 & 97.24 & 7.34 & \textcolor{red}{99.59} & 1.31 & 0.89 & 99.46 & 15.24 \\ 
 & QASP-DL & \textcolor{teal}{91.31} & 74.95 & 6.98 & 80.92 & 1.85 & \textcolor{teal}{95.30} & 36.46 & 4.24 & 86.08 & 2.80 & \textcolor{teal}{98.53} & 10.26 & 1.42 & 94.71 & 6.44 \\ 
\multirow{1}{*}{50\% HARD} & Oracle Nprobe & \textcolor{red}{83.17} & 163.25 & 10.87 & 50.97 & 3.69 & \textcolor{red}{91.50} & 66.01 & 6.24 & 64.36 & 6.78 & \textcolor{red}{97.26} & 13.13 & 2.45 & 84.89 & 14.39 \\ 
 & QASP-DL & \textcolor{teal}{90.44} & 48.17 & 5.63 & 79.14 & 6.07 & \textcolor{teal}{94.90} & 20.28 & 3.51 & 85.11 & 9.12 & \textcolor{teal}{98.41} & 3.92 & 1.31 & 94.40 & 16.50 \\
\bottomrule
\end{tabular}}
\end{table*}

\section{Performance Evaluation}\label{sec:experiments}
\subsection{Experimental Setup}
All experiments are performed on \textit{Intel(R) Xeon(R) Platinum 8175M CPU @ 2.50GHz \& 96 CPU cores}. We use Python 3.10 with Pytorch 1.13 for deep-learning and scikit-learn 1.2.2 for gradient boosting, regression and kmeans++.
\subsubsection{Vector Datasets}\label{sec:datasets}
We build QASP on seven datasets from a range of vector spaces, distance measures and application domains. 
SIFT1M \cite{sift}, 1 million 128-dimensional SIFT image descriptors; MNIST \cite{mnist}, 60k 784-dimensional vectors trained on handwritten digits; GIST1M \cite{gist}, 1 million 960-dimensional global color GIST descriptors. DEEP1B\_10M, 10 million subset of DEEP1B \cite{deep1b} with 96-dimensional deep learning embeddings; GLOVE-200 \cite{glove}, 1.2 million 200-dimensional word embeddings trained on Wikipedia; COCO-I2I \cite{mscoco}, 113k 512-dimensional vectors for image-to-image retrieval, COCO-T2I \cite{mscoco}, 113k 512-dimensional vectors for text-to-image retrieval. Each dataset has a query set with ground-truth $k=100$ nearest neighbors from exact search; we report $k=100$ throughout unless stated otherwise.

\subsubsection{Training, Validation and Test Sets}\label{sec:train_set}
We use a random sample of $\min(1000, \lfloor|Query\,Set|/10\rfloor)$ queries for generating offline training and validation sets. Model features, as discussed in section \ref{sec:features}, are obtained per query and centroid rank by searching up to $\min(200, \lfloor0.5 L\rfloor)$ probes. Ground truth is used to obtain recall target labels. The offline dataset is split 80:20 for training and validation. For testing, $1000$ queries are sampled from the remaining query set and features are obtained on-the-fly during search. All reported results use a fixed seed (1234); the reactive policy has minor run-to-run variance from parallel search. See Section \ref{sec:fit_analysis} for model performance and fit analysis.

\subsubsection{Model Details}
Deep Learning (DL) neural architecture of QASP, called QASP-DL, contains 3 hidden layers with hidden dimension of 18 with ReLU activation \cite{relu}. Each numerical input feature is mapped to a 4-dimensional embedding via a per-feature learnable linear (affine) layer followed by a ReLU nonlinearity, following the numerical-feature embedding framework for tabular deep learning of \cite{gorishniy2022embeddings}. We use batch normalization on input embeddings and dropout with rate 0.1 for regularization. Each DL model is trained for 100 epochs with a constant learning rate of 5e-3 and batch size of 512, though convergence is observed much earlier. The best model is selected using held-out validation set performance. Gradient boosting decision tree (GBDT) model variant of QASP, called QASP-GBDT, is trained using 100 trees with a maximum depth of 3 and a learning rate of 0.1. Polynomial variant of QASP, called QASP-LITE, is a lightweight and easy to adopt alternative in production environment, is fitted using Lasso regression with $\alpha$=1e-5 on sklearn \texttt{PolynomialFeatures} processing with degree 3.

\textit{Latency and Throughput}\label{sec:latency} - The DL model has a modest parameter count of $1513$ trainable parameters and a total memory footprint of just 6-7 KB, making it practical for deployment even in memory-constrained environments. For DL model, mean latency on CPU with batch size 32 is 0.26 ms/batch, throughput 124K queries/sec and P99 latency 0.33 ms using 4 threads. For GBDT, mean latency on CPU with batch size 32 is 0.199 ms/batch, throughput 161K queries/sec and P99 latency 0.55 ms. For QASP-LITE, mean latency on CPU with batch size 32 is 0.06 ms/batch, throughput 524K queries/sec and P99 latency 0.09 ms. All measurement is over 100 runs with 10 warm-up iterations.
\subsubsection{Index Details}
We first build flat IVF indices with nlist as ${scalar} \times \sqrt{n}$ where $n$ is the base size and use two $scalar$ values $[1, 4]$ to span across different index configurations per dataset. Distance metrics is either Squared Euclidean or Cosine Distance depending on the dataset and the clustering algorithm used is k-means++ \cite{kmeans++}. We also build a hierarchical index using two-level kmeans clustering
for scaling experiments to larger datasets of size ($> 10MM$). Each level of the index is clustered using fixed nlist as $\sqrt[3]{n}$ that gives equal sized partitions ideal for disk-based experiments. 

\subsection{Query Variability-Aware Evaluation}\label{sec:eval_framework}

The traditional method for evaluating vector similarity search results is to measure the deviation of observed average recall from recall target. Given observed average recall, $\bar{r} = \frac{1}{|Q|}
\sum_{q \in Q} r(q)$, where $r(q)$ is the observed recall for query $q$ and recall target $r^*$, traditional evaluation concerns only with $|\bar{r} - r^*|$. While useful as a summary statistic, it masks performance variations across queries and fails to capture user experience consistency.
To address this shortcoming, we propose a comprehensive Query Variability-aware Evaluation (QVE) framework with the following metrics to compare query variability even when $\bar{r} \approx r^*$ between search policies.

\begin{CompactItemize}
    \item \textbf{Recall Variance}: $\sigma^2 = \frac{1}{|Q|} \sum_{q \in Q} (r(q) - \bar{r})^2$ - Measures variance of the observed recall across test queries.
    
    \item \textbf{Absolute Query Deviation}: $\delta = \frac{1}{|Q|} \sum_{q \in Q} |r(q) - r^*|$ - Measures mean of absolute deviation from recall target.
    
    \item \textbf{Query Satisfaction Rate}: $S\% = \frac{|\{q \in Q : r(q) \geq r^*(1 - \gamma)\}|}{|Q|}$ - \% of queries satisfied by the search policy, i.e., achieve at least recall $r^*$. Since hitting $r^*$ exactly is rarely practical, we relax the threshold by a small margin $\gamma = 0.05$.
    \item \textbf{\% Data Accessed}: $A\% =$ $\frac{100}{|Q \times D|}$ $\sum_{q \in Q} |\{v \in D : d(v,q) \text{ is calculated} \}|$, \% of data accessed.
\end{CompactItemize}

\subsection{QASP versus Oracle Fixed Policy}\label{sec:main_results}
We first compare QASP to policies that search a fixed number of probes, which are industry standard where several open source libraries like Faiss \cite{faiss_guidelines}, Milvus \cite{milvus_guide}, and OpenSearch \cite{opensearch_guide} recommend tuning \texttt{nprobe} parameter for IVF for a given dataset and query workload. 
To construct the strongest such baseline, we fit Oracle Nprobe directly against the target: on the same training queries used for QASP, we sweep \texttt{nprobe} and select the smallest value whose mean recall reaches $r^*$. Because this is the best \textit{fixed} probe count achievable on QASP's own training distribution, Oracle Nprobe is a theoretical upper bound on what any non-adaptive policy can achieve, making it a strong baseline to QASP.
See Table \ref{table:main}, which contains comparison results averaged across the seven datasets.

Both Oracle Nprobe and QASP achieve similar $\bar{r}$'s close to the target. However, QASP outperforms Oracle Nprobe on QVE metrics achieving $56.2\%$ lower $\sigma^2$ and $33.3\%$ lower $\delta$, while
satisfying $10.9\%$ more queries and accessing almost the same data ($+2.1\%$) for $r^* = 90\%$. As $r^*$ increases, QASP accesses much fewer data than Oracle Nprobe to achieve the same recall with comparable
satisfaction rate. To further understand this behavior, we split the test dataset into two halves — \textit{50\% EASY} and \textit{50\% HARD}, using Local Relative Contrast feature. We observe that Oracle Nprobe
consistently over-reads in the \textit{50\% EASY} set and under-reads in the \textit{50\% HARD} set. QASP, being query adaptive, does the opposite by reading less in the \textit{50\% EASY} set and more in the
\textit{50\% HARD} set as intended. 

\subsection{Baselines}\label{sec:baselines}
We extend our comparison to proactive baselines from popular vector search libraries and prior work. Several libraries recommend tuning \texttt{nprobe} parameter for IVF \cite{milvus_guide,
faiss_guidelines, opensearch_guide} for which we use Oracle Nprobe. LanceDB \cite{lancedb_guide} suggests searching for a fixed \% of vectors for which we obtain Oracle Access \% from training dataset. SPANN
\cite{chen2021spann} is a popular clustering index with query-aware pruning that searches until cluster rank $j$ such that $Dist(q, c_j) \leq (1 + \epsilon) \times Dist(q, c_1)$. SQUASH \cite{oakley2025squash}
also involves a scalar multiple on distance to create a search mask. We obtain Oracle Distance Multiplier to represent these methods. Vexless \cite{vexless} uses a threshold on absolute distance, for which
we derive Oracle Distance. PCE-Net \cite{PceNet} is a learned baseline that predicts a single nprobe per query using neural encoders over the query vector and centroid distance distribution; we
evaluate them on the same IVF index for a fair comparison of the probing approaches.

Table \ref{tab:baseline} contains comparative results across the seven datasets. First, we observe that all the oracle baselines can be tuned to achieve $\bar{r} \approx r^*$ except Oracle Distance
\cite{vexless}. PCE-Net achieves high recall but at a steep cost: it reads 14.47\% of data—more than twice QASP-DL's $5.96\%$—as its asymmetric loss incentivizes over-probing. Its
variance ($\sigma^2 = 73.13$) and deviation ($\delta = 5.49$) also exceed all QASP variants. QASP consistently outperforms all baselines on QVE metrics while achieving recall target, satisfying more queries and accessing fewer vectors. We split test queries into easy and hard. We observe that fixed policies like Oracle Nprobe and Oracle Access\% over-read by
4.35\% for easy queries and under-read by 2.24\% for hard queries. Oracle Distance Multiplier is more adaptive; reading less in easy queries and more in hard queries. However, when compared to QASP models this
policy had 114\% higher variance. PCE-Net \cite{PceNet}, which predicts per query nprobe, exhibits a similar pattern: it over-reads on hard queries ($A\% = 14.88\%$) and under-reads on easy queries ($A\% = 14.07\%$). PCE-Net achieves recall target at higher $\sigma^2$ and $\delta$ compared to fixed policies like Oracle Nprobe and Oracle Access $\%$. Moreover, while it achieves a higher recall than QASP-DL, it reads more data ($14.47\%$ versus $5.96\%$) and satisfies fewer queries ($84.54\%$ versus $85.60\%$).

\begin{table}[t]
\centering
\caption{Comparison of QASP to proactive baselines at recall target 95\%. QASP outperforms Oracle policies averaged over seven datasets. QASP-DL outperforms other variants achieving least query deviation and data accessed while achieving recall target.}\label{tab:baseline}
\resizebox{\columnwidth}{!}{%
\begin{tabular}{ll|ccccc}
\toprule
\multirow{2}{*}{{Dataset}} & \multirow{2}{*}{{Search Policy}}
& \multicolumn{5}{c}{\(\mathbf{r^* = 95\%}\)}\\
\cmidrule(lr){3-7}
& &  \(\bar{r}\) & \(\sigma^2\downarrow\) & \(\delta\downarrow\) & \(S\%\uparrow\) & \(A\%\downarrow\) \\
\midrule
AVERAGE & Oracle Nprobe \cite{milvus_guide, faiss_guidelines, opensearch_guide} & \textcolor{teal}{95.04} & 49.58 & 5.25 & 80.80 & 7.06 \\
& Oracle Access \% \cite{lancedb_guide} & \textcolor{teal}{95.04} & 48.46 & 5.20 & 80.81 & 6.97 \\
& Oracle Distance Multiplier \cite{chen2021spann, oakley2025squash} & \textcolor{teal}{95.36} & 60.76 & 5.24 & 85.62 & 15.65 \\ 
& Oracle Distance \cite{vexless} & \textcolor{red}{93.04} & 287.93 & 8.76 & 84.08 & 18.62 \\ 
& PCE-Net \cite{PceNet} & \textcolor{teal}{95.23} & 73.13 & 5.49 & 84.54 & 14.47 \\
50\% EASY & Oracle Nprobe \cite{milvus_guide, faiss_guidelines, opensearch_guide} & \textcolor{red}{98.57} & 7.45 & 4.26 & 97.24 & 7.34 \\ 
& Oracle Access \% \cite{lancedb_guide} & \textcolor{red}{98.34} & 9.22 & 4.25 & 95.97 & 6.96 \\ 
& Oracle Distance Multiplier \cite{chen2021spann, oakley2025squash} & \textcolor{red}{93.66} & 87.28 & 6.12 & 79.08 & 8.31 \\
& Oracle Distance \cite{vexless} & \textcolor{red}{98.79} & 26.58 & 4.94 & 97.43 & 20.59 \\
& PCE-Net \cite{PceNet} & \textcolor{red}{97.84} & 41.85 & 4.75 & 95.70 & 14.07 \\ 
50\% HARD & Oracle Nprobe \cite{milvus_guide, faiss_guidelines, opensearch_guide} & \textcolor{red}{91.50} & 66.01 & 6.24 & 64.36 & 6.78 \\  
& Oracle Access \% \cite{lancedb_guide} & \textcolor{red}{91.73} & 65.17 & 6.15 & 65.65 & 6.98 \\ 
& Oracle Distance Multiplier \cite{chen2021spann, oakley2025squash} & \textcolor{red}{97.06} & 25.37 & 4.36 & 92.16 & 23.00 \\ 
& Oracle Distance \cite{vexless} & \textcolor{red}{87.30} & 463.19 & 12.58 & 70.73 & 16.65 \\
& PCE-Net \cite{PceNet} & \textcolor{red}{92.63} & 85.52 & 6.23 & 73.39 & 14.88 \\
\midrule
AVERAGE & QASP-DL & \textcolor{teal}{95.09} & 28.45 & \textbf{3.87} & 85.60 & \textbf{5.96} \\ 
& QASP-GBDT & \textcolor{teal}{95.13} & \textbf{26.72} & 3.92 & 84.17 & 6.14 \\  
& QASP-LITE & \textcolor{teal}{95.40} & 38.13 & 4.18 & \textbf{86.56} & 7.13 \\
50\% EASY & QASP-DL & \textcolor{teal}{95.30} & 36.46 & 4.24 & 86.08 & 2.80 \\ 
& QASP-GBDT & \textcolor{teal}{95.60} & 28.29 & 4.11 & 85.88 & 2.65 \\ 
& QASP-LITE & \textcolor{teal}{95.00} & 47.72 & 4.69 & 83.74 & 3.10 \\
50\% HARD & QASP-DL & \textcolor{teal}{94.90} & 20.28 & 3.51 & 85.11 & 9.12 \\ 
& QASP-GBDT & \textcolor{teal}{94.66} & 23.99 & 3.73 & 82.45 & 9.64 \\ 
& QASP-LITE & \textcolor{teal}{95.77} & 28.03 & 3.66 & 89.39 & 11.17 \\
\bottomrule
\end{tabular}}
\end{table}

\subsection{Deployment Efficiency}

\begin{figure}[b]
\centering
\begin{subfigure}{0.48 \columnwidth}
    \centering
    \includegraphics[width=\linewidth]{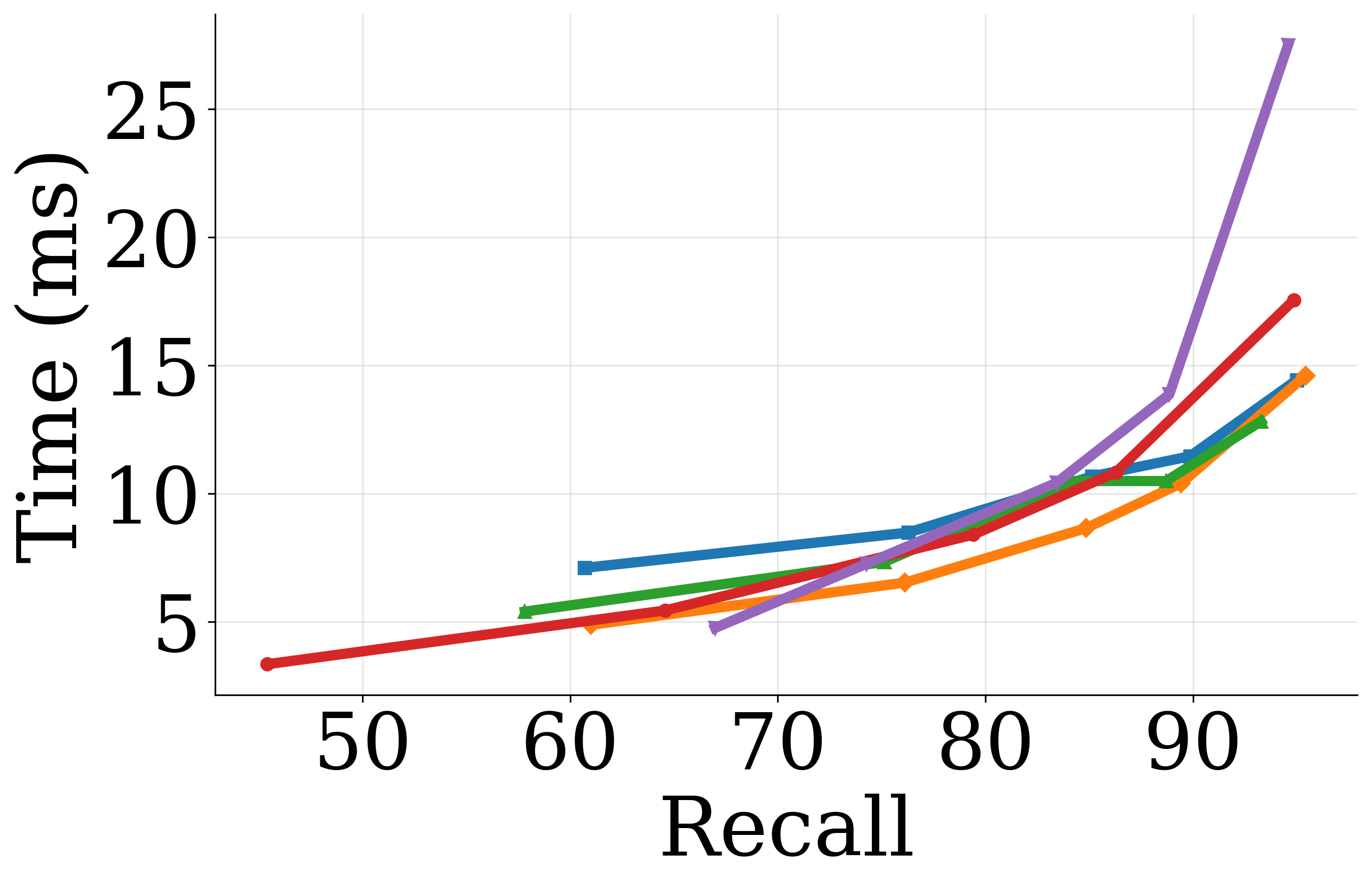}
    \caption{In-Memory}
\end{subfigure}
\begin{subfigure}{0.48 \columnwidth}
    \centering
    \includegraphics[width=\linewidth]{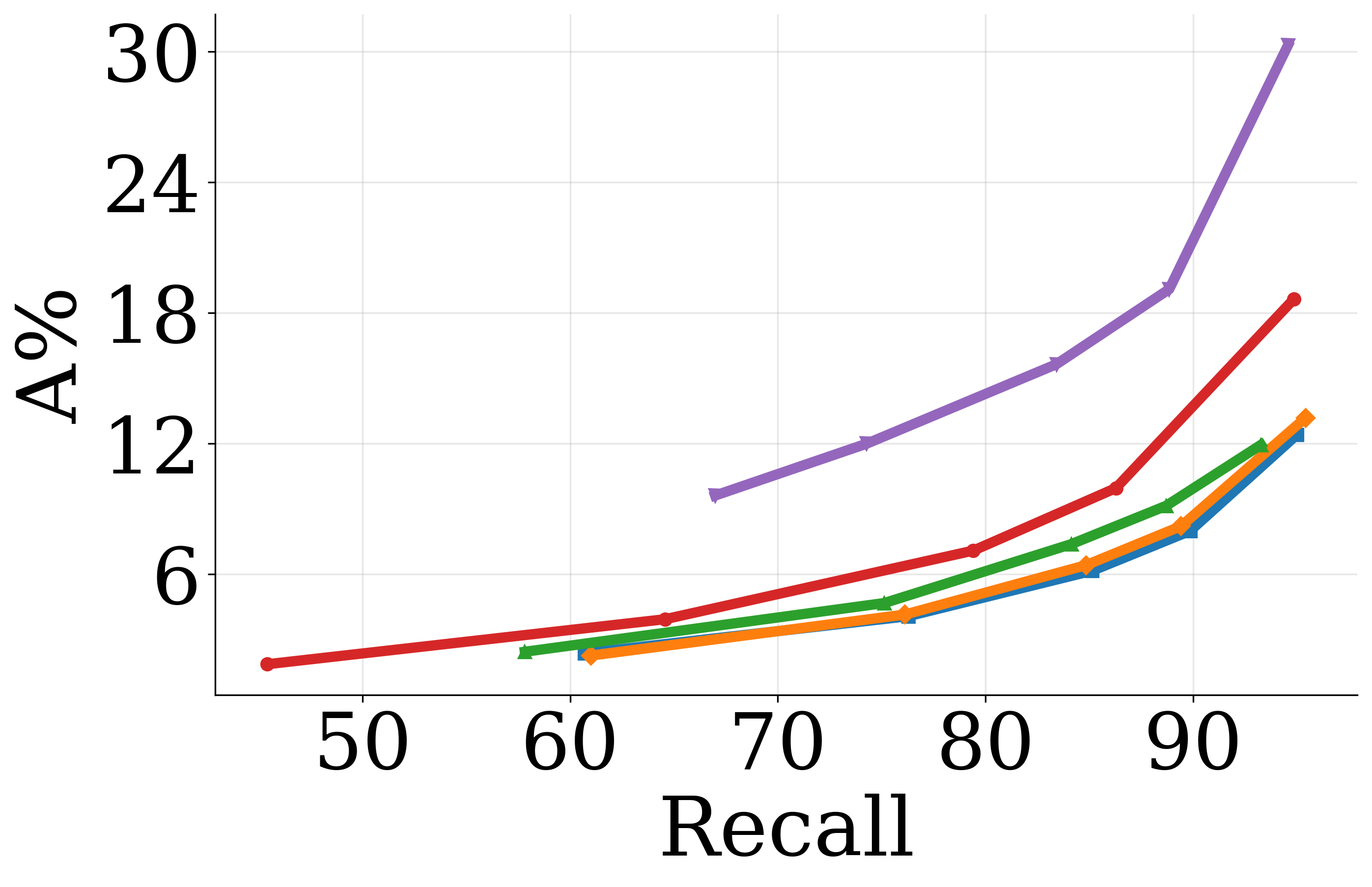}
    \caption{Disk-Based}
\end{subfigure}

\caption{Search time (ms) vs. recall (in-memory) and data accessed vs. recall (disk-based). QASP achieves recall target with less search cost than Oracle Nprobe and PCE-Net, with gains further widening at higher recall targets. Legend:
\textcolor{mplblue}{QASP-DL}, \textcolor{orange}{QASP-GBDT}, \textcolor{mplgreen}{QASP-LITE}, \textcolor{red}{Oracle Nprobe}, \textcolor{violet}{PCE-Net}}
\label{fig:search_efficiency}
\end{figure}

Figure \ref{fig:search_efficiency} evaluates QASP's deployment efficiency in both in-memory and disk-based scenarios. We include Oracle Nprobe and PCE-Net as strong baselines for both latency and recall. For in-memory deployment, we measure total runtime (milliseconds) required to satisfy $95\%$ of queries. QASP incurs a single inference overhead plus search time, while Oracle Nprobe only incurs search time. QASP's constant inference overhead becomes proportionally smaller compared to search time at higher recall targets and with larger datasets (see Section~\ref{sec:scaling}). We observe diminishing recall returns at higher recalls with Oracle Nprobe compared to QASP models, eventually offsetting inference cost.

For disk-based settings, QASP's benefits are even more pronounced as I/O costs dominate search time. The single inference becomes negligible compared to I/O costs, and QASP directly reduces disk operations. QASP variants consistently read 20-40\% less data than Oracle Nprobe across all recall levels, with this efficiency gap widening as recall targets increase (Figure~\ref{fig:search_efficiency}). This behavior is supported by the theoretical dominance condition illustrated in Figure~\ref{fig:scan_cost}.

\subsection{Domain Adaptation}\label{sec:generalization_results}
\subsubsection{Cross-Dataset Generalization}
We evaluate QASP's generalization to unseen datasets using two QVE metrics ($\bar{r}$ and $\delta$) under two transfer settings: zero-shot (direct application) and few-shot (fine-tuning with $1\%$ of target data; see Section~\ref{sec:generalization}). As shown in Table~\ref{tab:generalization_datasets}, QASP-DL often generalizes well zero-shot, e.g., GIST1M$\to$SIFT1M achieves $\bar{r}=98.60$ without tuning. Where gaps exist (e.g., SIFT1M$\to$GIST1M), few-shot fine-tuning with minimal data closes them effectively. Training on diverse sources further helps: a model trained on SIFT1M+MNIST combined outperforms individual models on GIST1M in the zero-shot setting.

\begin{table}[h]
\centering
\caption{Cross-dataset generalization of QASP for recall target 95\%. QASP-DL models generalize well under zero-shot and achieve recall close to target, while few-shot training bridges the gap on average with fully trained.
QASP models can be fine-tuned with just 1\% queries from target domain.}

\label{tab:generalization_datasets}
\resizebox{\columnwidth}{!}{%
\begin{tabular}{ll|cc|cc|cc}
\toprule
\textbf{Source} & \textbf{Target} & \multicolumn{2}{c|}{\textbf{Zero-Shot}} & \multicolumn{2}{c|}{\textbf{Few-Shot}} & \multicolumn{2}{c}{\textbf{Fully Trained}} \\
 &  & \(\bar{r}\) & \(\delta\downarrow\) & \(\bar{r}\) & \(\delta \downarrow\) & \(\bar{r}\) & \(\delta\downarrow\) \\
\midrule
GIST1M & SIFT1M  & 98.60 & 4.01 & 93.70 & 4.23 & 95.47 & 3.44\\
MNIST & GIST1M  & 95.50 & 3.97 & 92.00 & 5.09 & 92.30 & 4.56\\
SIFT1M & GIST1M  & 26.80 & 68.22 & 92.90 & 4.36 & 92.30 & 4.56\\
SIFT1M + MNIST & GIST1M  & 93.70 & 4.57 & 94.60 & 3.87 & 92.30 & 4.56\\
COCOI2I & GLOVE200  & 93.10 & 5.52 & 94.90 & 3.84 & 95.23 & 3.46 \\
GLOVE200 & COCOI2I  & 91.80 & 6.30 & 96.80 & 3.63 & 95.83 & 3.98 \\
\textbf{AVERAGE} & - & \textbf{83.25} & \textbf{15.76} & \textbf{94.15} & \textbf{4.17} & \textbf{93.91} & \textbf{4.09} \\
\bottomrule
\end{tabular}}
\end{table}

\begin{table}[b]
\centering
\caption{Cross-index generalization of QASP on QVE metrics for recall target = 95\%. QASP models trained on one index configuration (source) generalize well zero-shot to unseen configurations (target) reducing the need of fine-tuning.}
\label{tab:generalization_indices}
\resizebox{\columnwidth}{!}{
\begin{tabular}{lcc|cc|cc}
\toprule
\textbf{Dataset} & \textbf{Source Config} & \textbf{Target Config} &
\multicolumn{2}{c|}{\textbf{Zero-Shot}} &
\multicolumn{2}{c}{\textbf{Fully Trained}} \\
 &  & & \(\bar{r}\) & \(\delta\downarrow\) & \(\bar{r}\) & \(\delta\downarrow\)  \\
\midrule
SIFT1M  & IVF(L=1000) & IVF(L=4000)     & 85.00 & 10.83 & 93.50 & 4.08 \\
SIFT1M  & IVF(L=4000) & IVF(L=1000)     & 96.30 & 3.30 & 95.47 & 3.44 \\
GIST1M  & IVF(L=1000) & IVF(L=4000)     & 91.50 & 5.17 & 91.40 & 5.17 \\
GIST1M  & IVF(L=4000) & IVF(L=1000)     & 92.30 & 4.56 & 92.30 & 4.56 \\
COCOI2I  & IVF(L=336) & IVF(L=1346)     & 87.80 & 9.28 & 95.10 & 4.08 \\
COCOI2I  & IVF(L=1346) & IVF(L=336)     & 97.20 & 3.74 & 95.83 & 3.98 \\
GLOVE200  & IVF(L=1087) & IVF(L=4531)   & 94.70 & 3.79 & 94.70 & 3.79 \\
GLOVE200  & IVF(L=4531) & IVF(L=1087)   & 95.10 & 3.59 & 95.23 & 3.46 \\
\textbf{AVERAGE} & - & - & \textbf{93.84} & \textbf{4.78} & \textbf{94.29} & \textbf{4.07} \\
\bottomrule
\end{tabular}}
\vspace{-5pt}
\end{table}

\subsubsection{Cross-Index Configuration Generalization}
We investigate the robustness of QASP models to varying index configurations by testing how well they perform on the same dataset but with a different number of clusters $L$. As observed in Table \ref{tab:generalization_indices}, QASP-DL models generalize well on unseen index configurations without fine-tuning. Domain adaptation methods remain applicable where few-shot learning can bridge the gap. However, zero-shot application of QASP models to unseen indices produces competitive performance on average.

\subsection{Reactive Complement}\label{sec:experiments_reactive}
We evaluate the performance of QASP's reactive complement described in Section~\ref{sec:reactive_complement}, which guides policy behavior at search. See Algorithm \ref{alg:reactive_qasp} for the pseudocode of the reactive complement policy based on observed discovery rates. The implementation parameters include a smoothing factor $\alpha = 0.3$, $L_{\delta}=3$, margins $\epsilon_1=0.05, \epsilon_2=4.0$ and a threshold of 0.35 on $\delta_i$ to capture positive and negative rates. For baselines, we use LAET \cite{laet} which pauses search at a fixed nprobe and predicts when to stop, and DARTH \cite{chatzakis2025darth} which predicts recall at regular intervals to terminate the search early.
We tune LAET's \emph{multiplier} parameter to attain recall target of 95\%. For DARTH, we set the \emph{ipi} and \emph{mpi} parameters per target recall from its training statistics, and retrain DARTH's predictors for $k=100$ (the released models use $k=50$).
Table~\ref{tab:reactive} shows results when all mechanisms are applied on IVF index. Both QASP and QASP + Reactive consistently outperform LAET and DARTH on all three datasets with QASP + Reactive further boosting QASP performance.

\begin{algorithm}[h]
\caption{Reactive Adaptive Search for a Single Query}
\label{alg:reactive_qasp}
\begin{algorithmic}[1]
\Require Query $q$, nearest cluster ordering $\mathcal{C}_q$, planned probes $P$, recall target $r^*$
\Require Recall estimates $\{\hat{r}_i\}_{i=1}^{|\mathcal{C}_q|}$ from QASP model

\Require Constants: $\epsilon_{1}/\epsilon_{2}$ (margins), $\delta$ (deviation threshold), $\alpha$ (smoothing factor), $L_{\delta}$ (consecutive deviation limit)
\Ensure Top-$k$ results and actual probes used

\State $P_{\min} \leftarrow \max\!\left(1,\; \lfloor P(1-\epsilon_1)\rfloor - L_{\delta} \right)$
\State $P_{\max} \leftarrow \min\!\left(|\mathcal{C}_q|,\; \lceil P(1+\epsilon_2)\rceil - L_{\delta}\right)$

\State Initialize smoothed discovery rate $SDR \leftarrow 0$
\State Initialize counters $c^{-} \leftarrow 0$, $c^{+} \leftarrow 0$
\State Initialize previous neighbors set $\mathcal{R}_{\text{prev}} \leftarrow \emptyset$

\For{$p = P_{\min}$ \textbf{to} $P_{\max}$}
\State $\mathcal{R}_p \leftarrow \textsc{SearchOneCluster}(q, \mathcal{C}_q^{p}, \mathcal{R}_{\text{prev}})$

\If{$p = P_{\min}$} \Comment{Skip comparison for first probe}
    \State $\mathcal{R}_{\text{prev}} \leftarrow \mathcal{R}_p$
    \State \textbf{continue}
\EndIf

\State $\text{DR}_p \leftarrow {|\mathcal{R}_p \setminus \mathcal{R}_{\text{prev}}|}/{k}$ \Comment{True discovery rate}
\State $\widetilde{\text{DR}}_p \leftarrow \max(1e^{-3},\; \hat{r}_p - \hat{r}_{p-1})$ \Comment{Estimated discovery rate}

\State $SDR \leftarrow \alpha \cdot \text{DR}_p + (1-\alpha)\cdot SDR$
\State $\Delta_p \leftarrow {(SDR - \widetilde{\text{DR}}_p)}/{\widetilde{\text{DR}}_p}$

\If{$\Delta_p < -\delta$}
    \State $c^{-} \leftarrow c^{-} + 1$, \quad $c^{+} \leftarrow 0$ \Comment{Negative discovery}
\ElsIf{$\Delta_p > \delta$}
    \State $c^{+} \leftarrow c^{+} + 1$, \quad $c^{-} \leftarrow 0$ \Comment{Positive discovery}
\Else
    \State $c^{-}, c^{+} \leftarrow 0$
\EndIf

\If{$c^{-} \geq L_{\delta}$ \textbf{and} $\hat{r}_p \geq r*(1-\epsilon)$}
    \State \textbf{break} \Comment{Over-search detected}
\EndIf

\If{$p = P$ \textbf{and} $c^{-} = 0$ \textbf{and} $c^{+} = 0$}
    \State \textbf{break} \Comment{Consistent with estimation}
\EndIf

\State $\mathcal{R}_{\text{prev}} \leftarrow \mathcal{R}_p$
\EndFor

\State \Return $\mathcal{R}_p$
\end{algorithmic}
\end{algorithm}

\begin{table}[h]
\centering
\caption{Comparison of QASP to reactive baselines at recall target 95\%. QASP outperforms both baselines and reactive complement improves QVE metrics at higher average recall.}

\label{tab:reactive}
\resizebox{\columnwidth}{!}{
\begin{tabular}{ll|ccccc}
\toprule
 \multirow{2}{*}{{Dataset}} & \multirow{2}{*}{{Search Policy}}
& \multicolumn{5}{c}{\(\mathbf{r^* = 95\%}\)}\\
\cmidrule(lr){3-7}
& &  \(\bar{r}\) & \(\sigma^2\downarrow\) & \(\delta\downarrow\) & \(S\%\uparrow\) & \(A\%\downarrow\) \\
\midrule
 SIFT1M & LAET \cite{laet} & 94.90 & 109.17 & 6.18 & 84.70 & 4.29 \\
  & DARTH \cite{chatzakis2025darth} & 94.20 & 29.45 & 4.13 & 79.50 & 3.19 \\
  & QASP-DL & 95.40 & 18.62 & 3.37 & 87.90 & 2.84 \\
 & QASP + Reactive & 96.90 & 13.04 & 3.51 & 93.30 & 3.45 \\
 GIST1M & LAET \cite{laet} & 94.80 & 119.20 & 5.61 & 87.22 & 10.52 \\
 & DARTH \cite{chatzakis2025darth} & 91.80 & 48.33 & 5.63 & 64.67 & 9.61 \\
  & QASP-DL & 92.30 & 47.89 & 4.72 & 72.67 & 6.69 \\
 & QASP + Reactive & 93.50 & 36.41 & 4.19 & 76.44 & 7.09 \\
 DEEP1B\_10M & LAET \cite{laet} & 94.90 & 43.81 & 4.81 & 82.10 & 0.85 \\
 & DARTH \cite{chatzakis2025darth} & 95.80 & 27.15 & 4.23 & 85.80 & 1.36 \\
  & QASP-DL & 95.20 & 27.42 & 3.80 & 86.70 & 0.88 \\
 & QASP + Reactive & 96.80 & 17.57 & 3.78 & 92.80 & 1.08 \\
 \midrule
 AVERAGE & LAET \cite{laet} & \textcolor{teal}{94.87} & 90.73 & 5.53 & 84.67 & 5.22 \\
& DARTH \cite{chatzakis2025darth} & \textcolor{red}{93.93} & 34.98 & 4.66 & 76.66 & 4.72 \\
  & QASP-DL & \textcolor{teal}{94.30} & 31.31 & 3.96 & 82.42 & \textbf{3.47} \\
 & QASP + Reactive & \textcolor{teal}{95.73} & \textbf{22.34} & \textbf{3.83} & \textbf{87.51} & 3.87 \\
\bottomrule
\end{tabular}}
\end{table}

\subsection{Scaling to Hierarchical Index}\label{sec:scaling}
We evaluate hierarchical scaling (Section \ref{sec:other_indices}) by applying a QASP-DL model trained on Deep1B-10M to larger subsets of both SIFT1B and Deep1B (10M–100M). For each subset, we construct a two-level index with \texttt{nlist1} = \texttt{nlist2} = $\sqrt[3]{n}$, yielding balanced partitions suitable for disk-based search. We compare against \textit{Oracle Nprobe}, which optimizes \texttt{nprobe1} via binary search with \texttt{nprobe2=nlist2}.
QASP operates level-wise, predicting adaptive \texttt{nprobe1} from first-level features and selectively probing second-level partitions based on predicted recall contribution. As shown in Figure~\ref{fig:two_level_scalability}, this reduces data access by $\geq$80\% at 99\% recall consistently across both datasets (SIFT1B: 82--87\%, Deep1B: 80--84\%), demonstrating that QASP's efficiency gains generalize across data distributions and scale without
retraining.

\begin{figure}[t]
\centering
\includegraphics[width=\columnwidth]{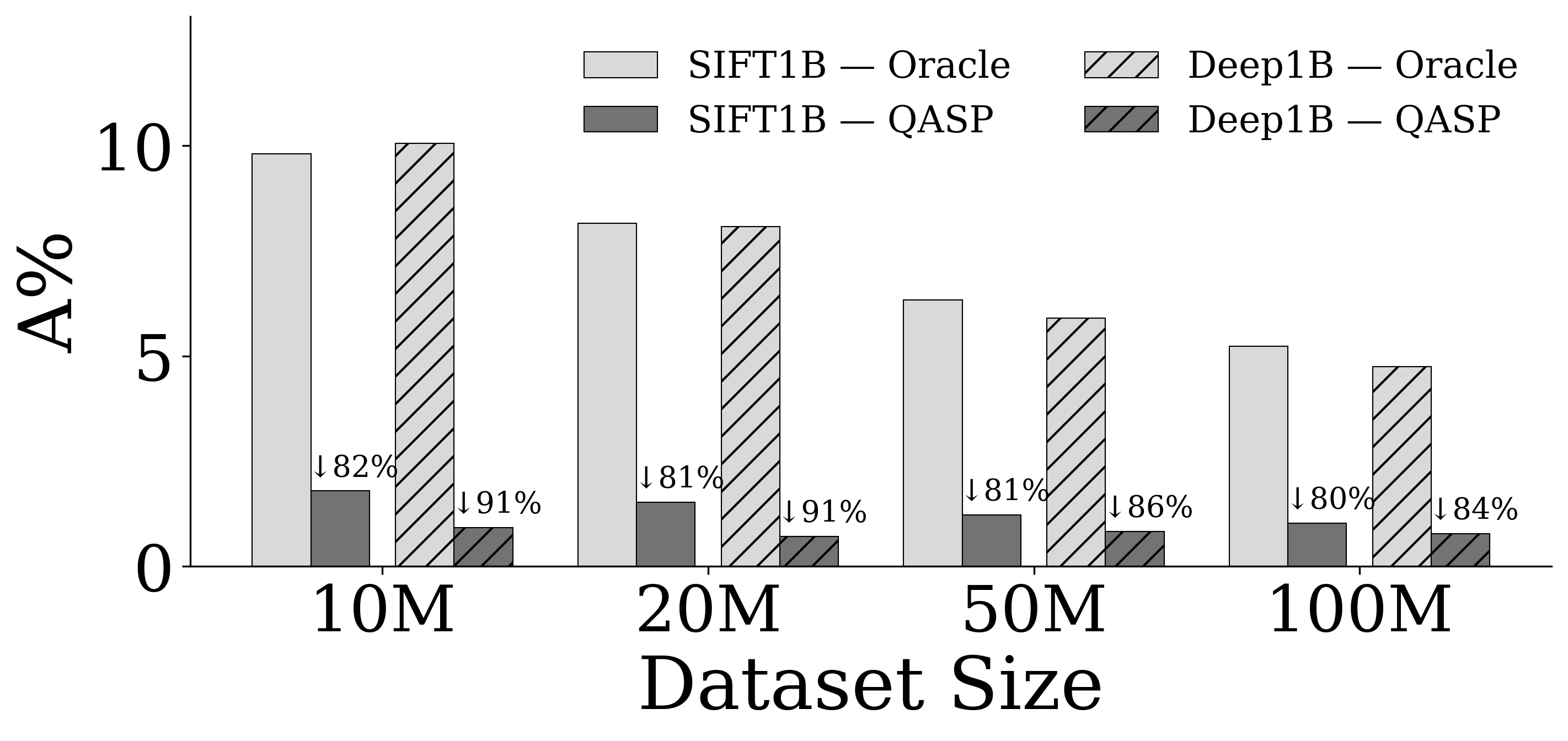}
\caption{$A\%$ for SIFT1B and Deep1B subsets (10M--100M) on a two-level hierarchical index at 99\% recall target. QASP is trained on a 10M subset and applied at inference without retraining. QASP consistently yields $\geq$80\% reduction in $A\%$.}
\label{fig:two_level_scalability}
\end{figure}

\section{Regression Fit Analysis}\label{sec:fit_analysis}

\subsection{Training Error Analysis}
\begin{figure}[h]
\centering
\begin{subfigure}{0.22\columnwidth}
    \includegraphics[width=\linewidth]{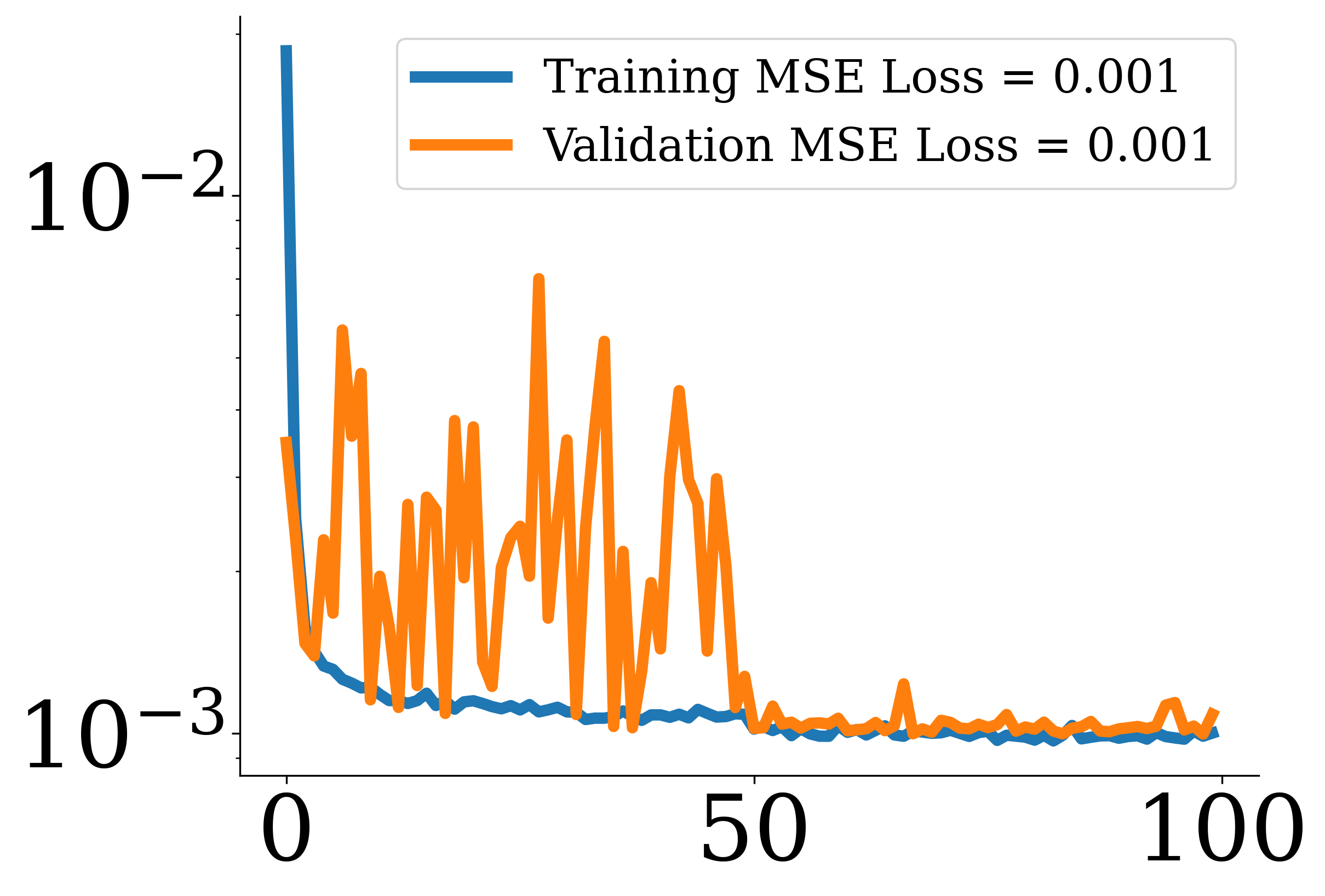}
    \caption{MNIST}
\end{subfigure}
\hfill
\begin{subfigure}{0.22\columnwidth}
    \includegraphics[width=\linewidth]{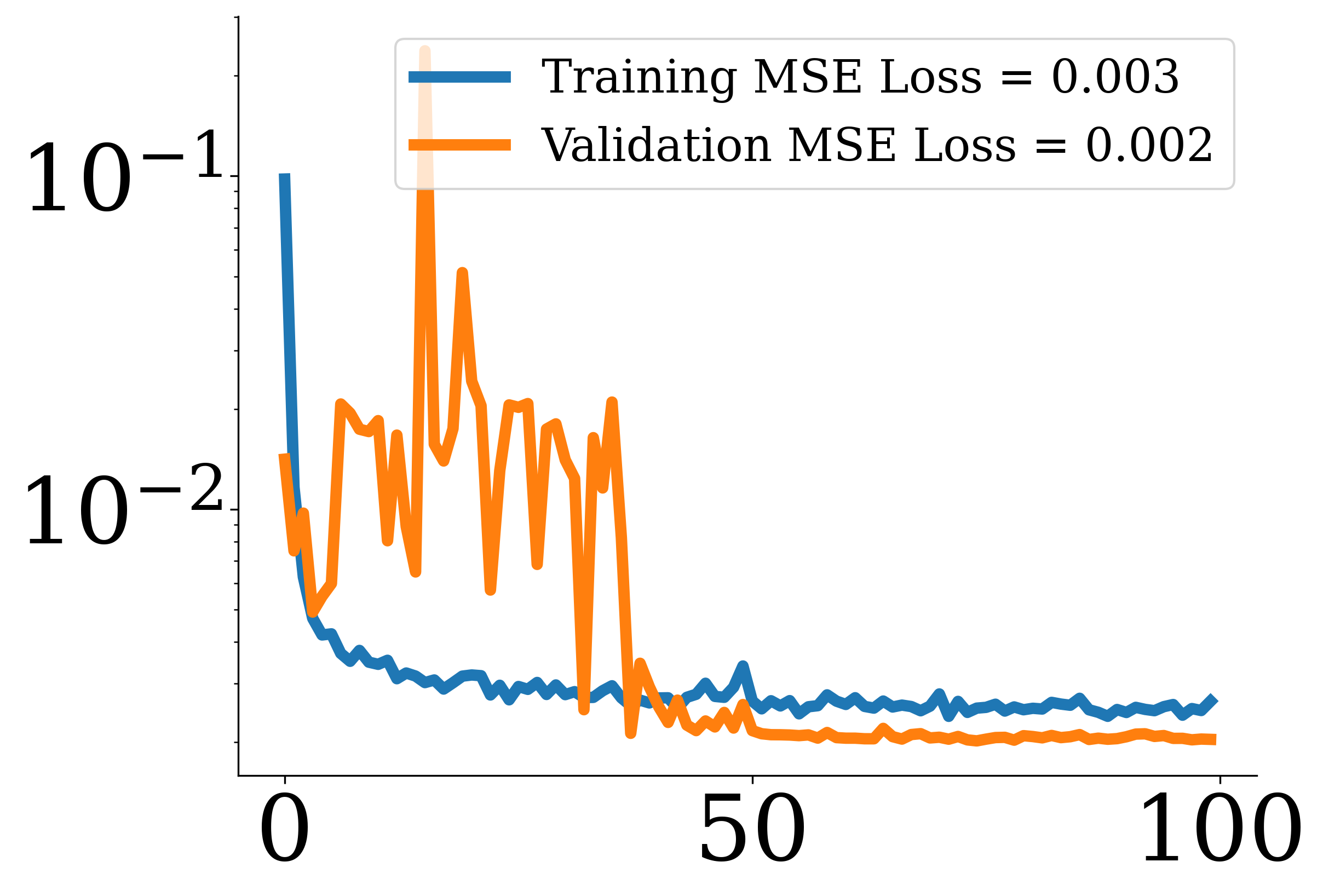}
    \caption{GIST}
\end{subfigure}
\hfill
\begin{subfigure}{0.22\columnwidth}
    \includegraphics[width=\linewidth]{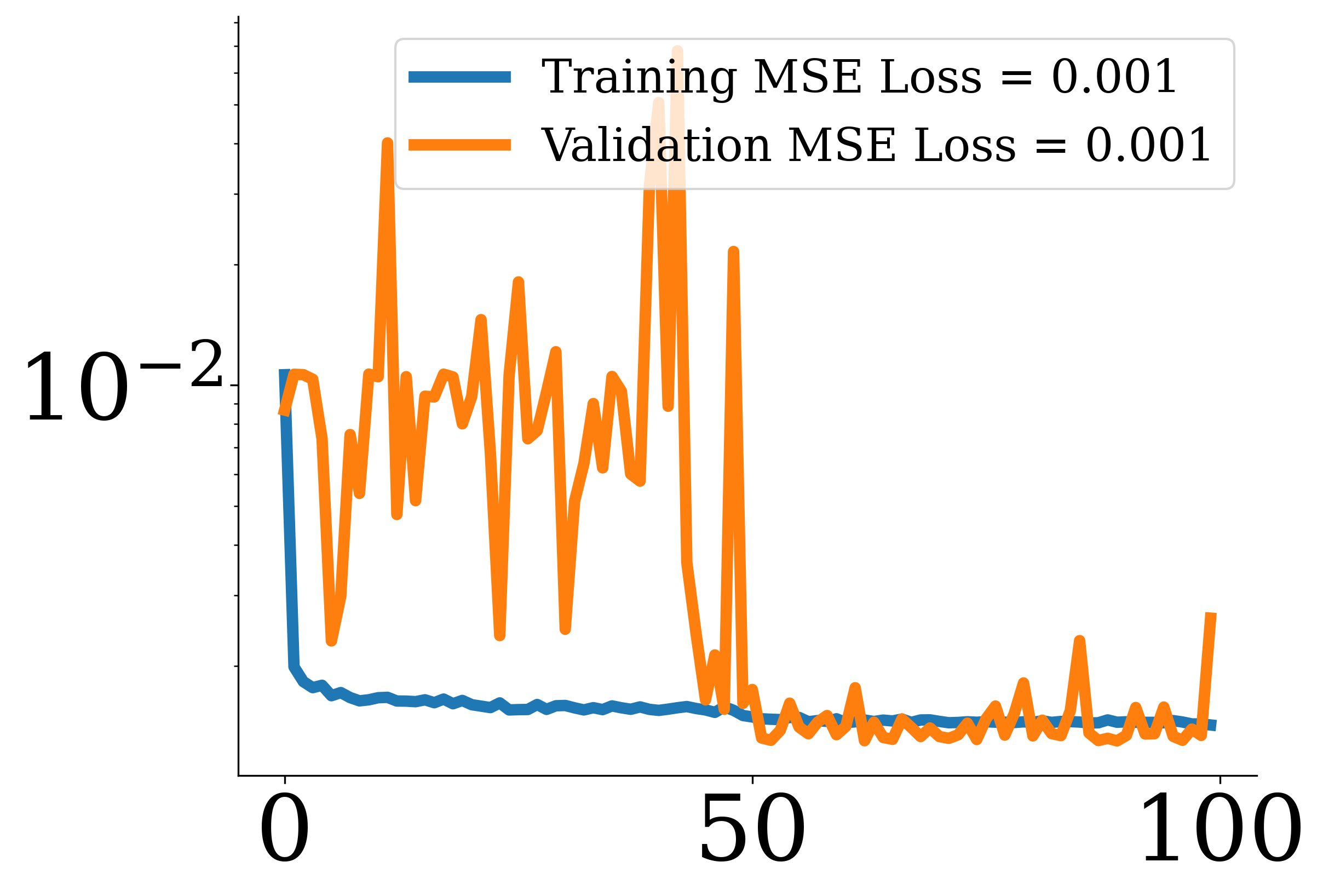}
    \caption{SIFT1M}
\end{subfigure}
\hfill
\begin{subfigure}{0.22\columnwidth}
    \includegraphics[width=\linewidth]{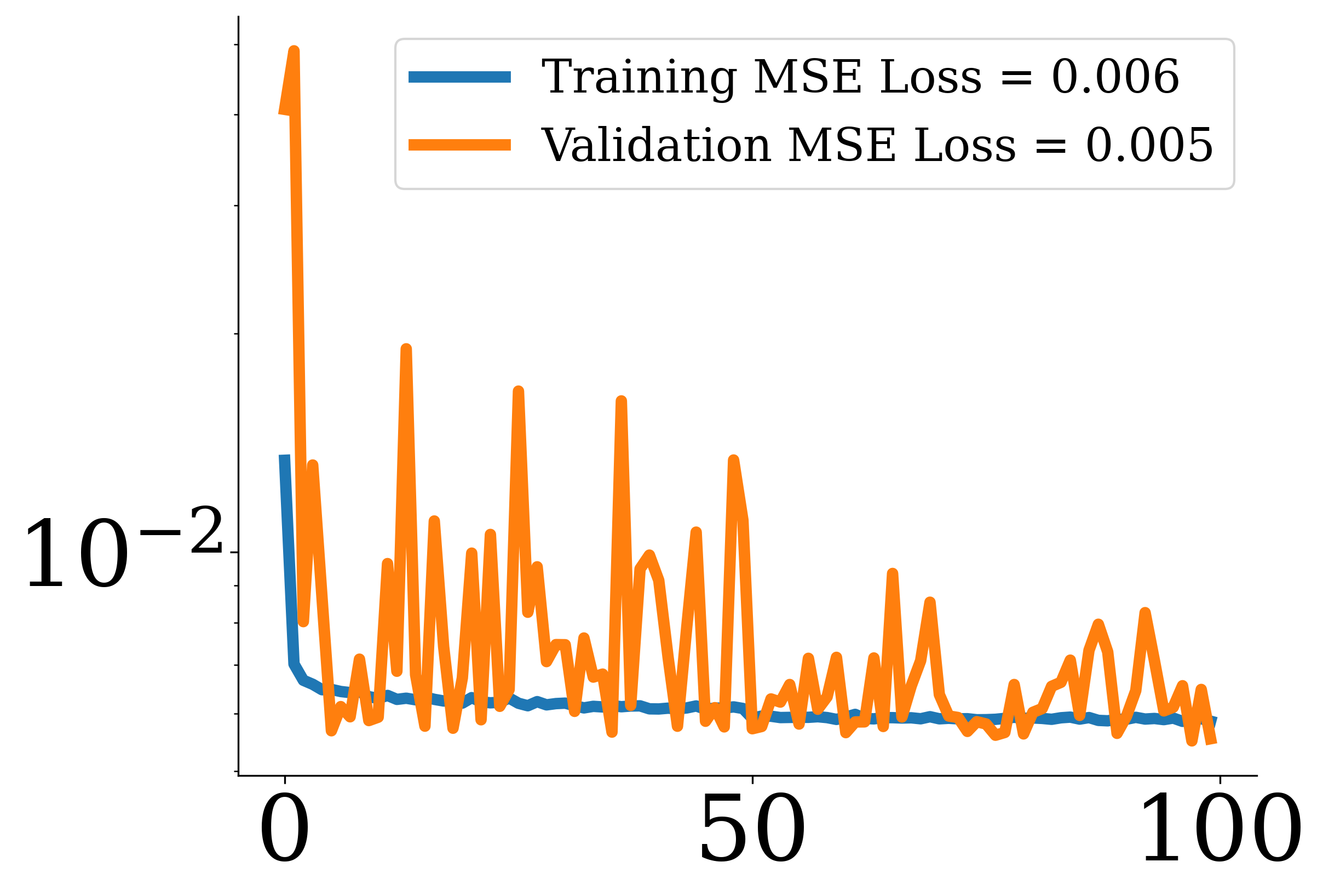}
    \caption{GLOVE-200}
\end{subfigure}
\hfill
\begin{subfigure}{0.22\columnwidth}
    \includegraphics[width=\linewidth]{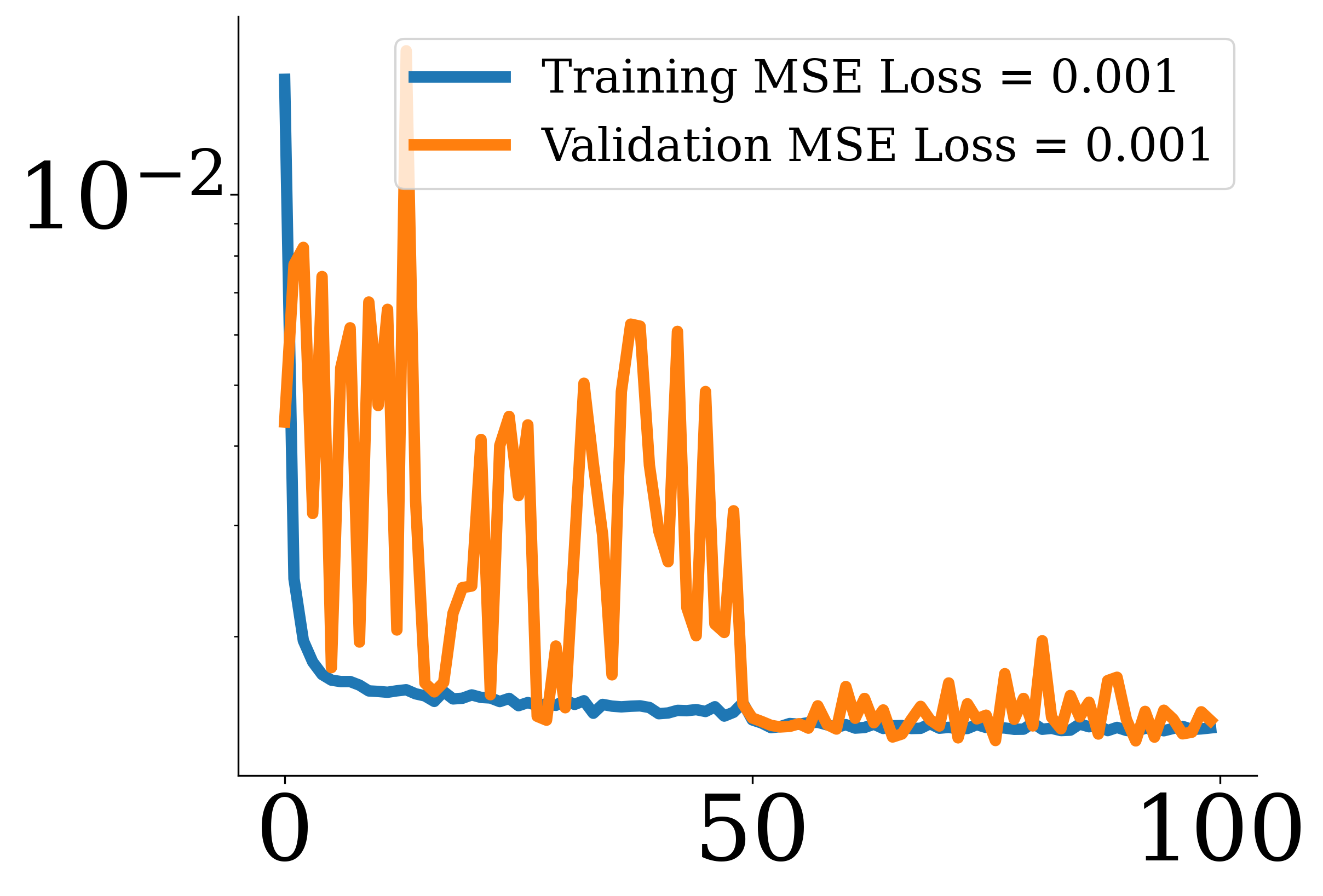}
    \caption{DEEP1B\_10M}
\end{subfigure}
\begin{subfigure}{0.22\columnwidth}
    \includegraphics[width=\linewidth]{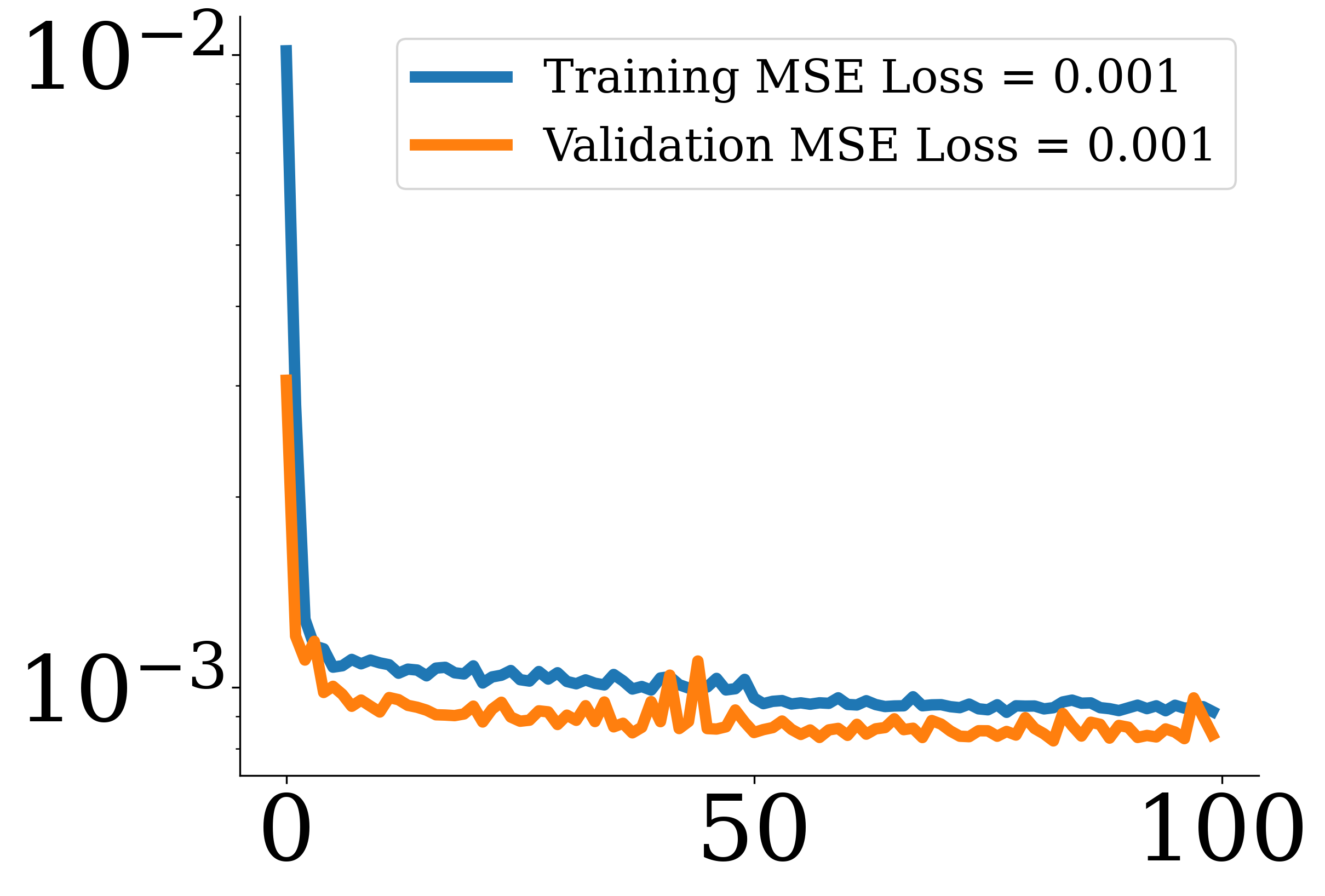}
    \caption{COCOI2I}
\end{subfigure}
\begin{subfigure}{0.22\columnwidth}
    \includegraphics[width=\linewidth]{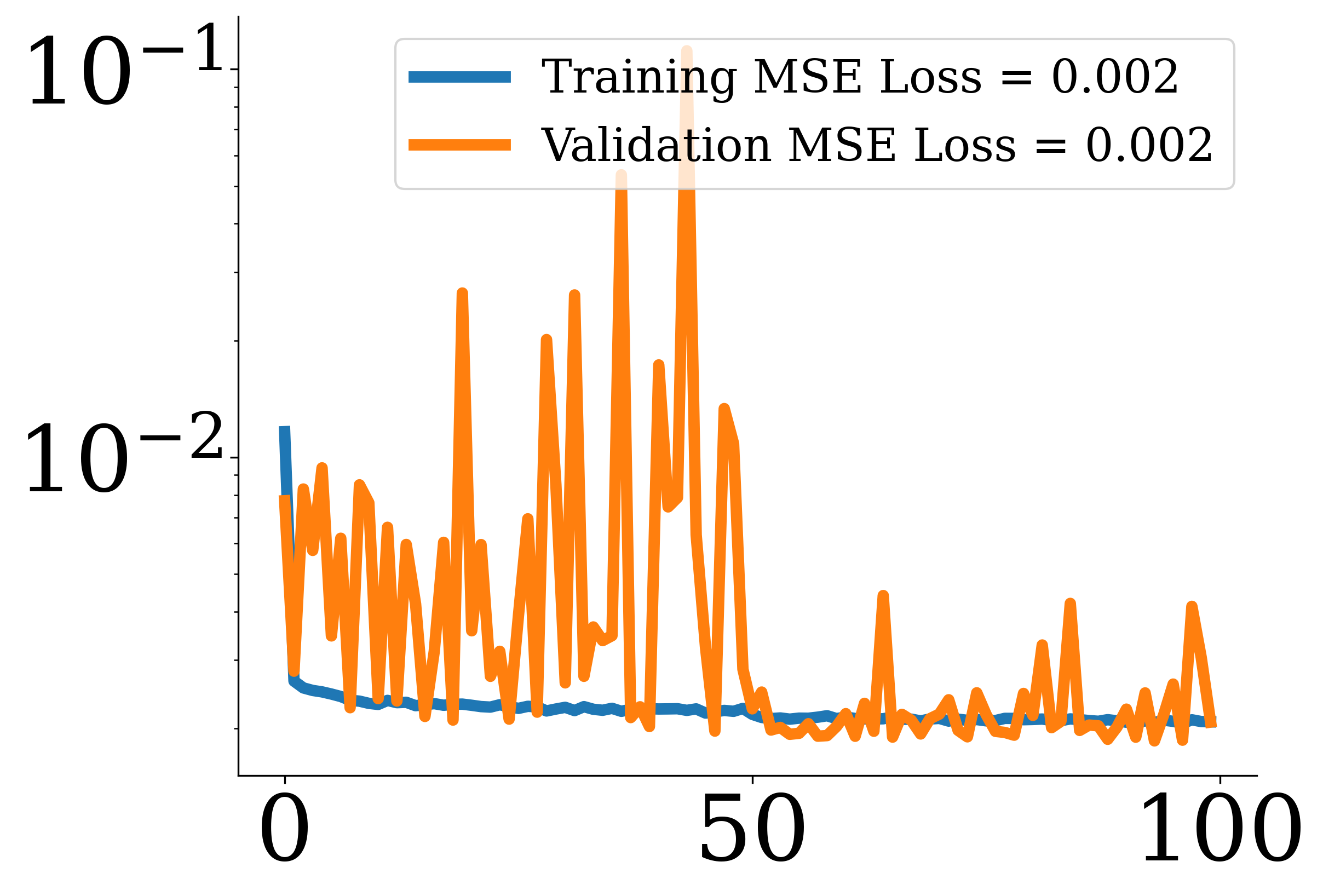}
    \caption{COCOT2I}
\end{subfigure}
\caption{Training loss (MSE) over epochs}
\label{fig:training_error_combined}
\end{figure}
Figure~\ref{fig:training_error_combined} presents training error curves on the seven datasets. All models converge rapidly, with most error reduction occurring within the first 20--30 epochs. Training and validation curves follow similar trajectories with minimal gaps, indicating good generalization without overfitting. Dataset-specific differences in convergence speed and final error levels reflect the varying complexity of the recall prediction task. Overall, QASP trains efficiently across diverse datasets, enabling rapid adaptation to new index configurations and query distributions in practice.

\subsection{Prediction Error Analysis}
Figure~\ref{fig:prediction_error_combined} shows actual versus predicted recall for DL and GBDT models (we omit Lite for brevity). Both models achieve MSE of 0.002 and similar R$^2$ with DL slightly better. Errors are balanced around the diagonal with a slight bias toward over-prediction—preferable for meeting recall targets. Prediction accuracy remains high across the entire recall range.

\begin{figure}[h]
\centering
\begin{subfigure}{0.23\columnwidth}
    \includegraphics[width=\linewidth]{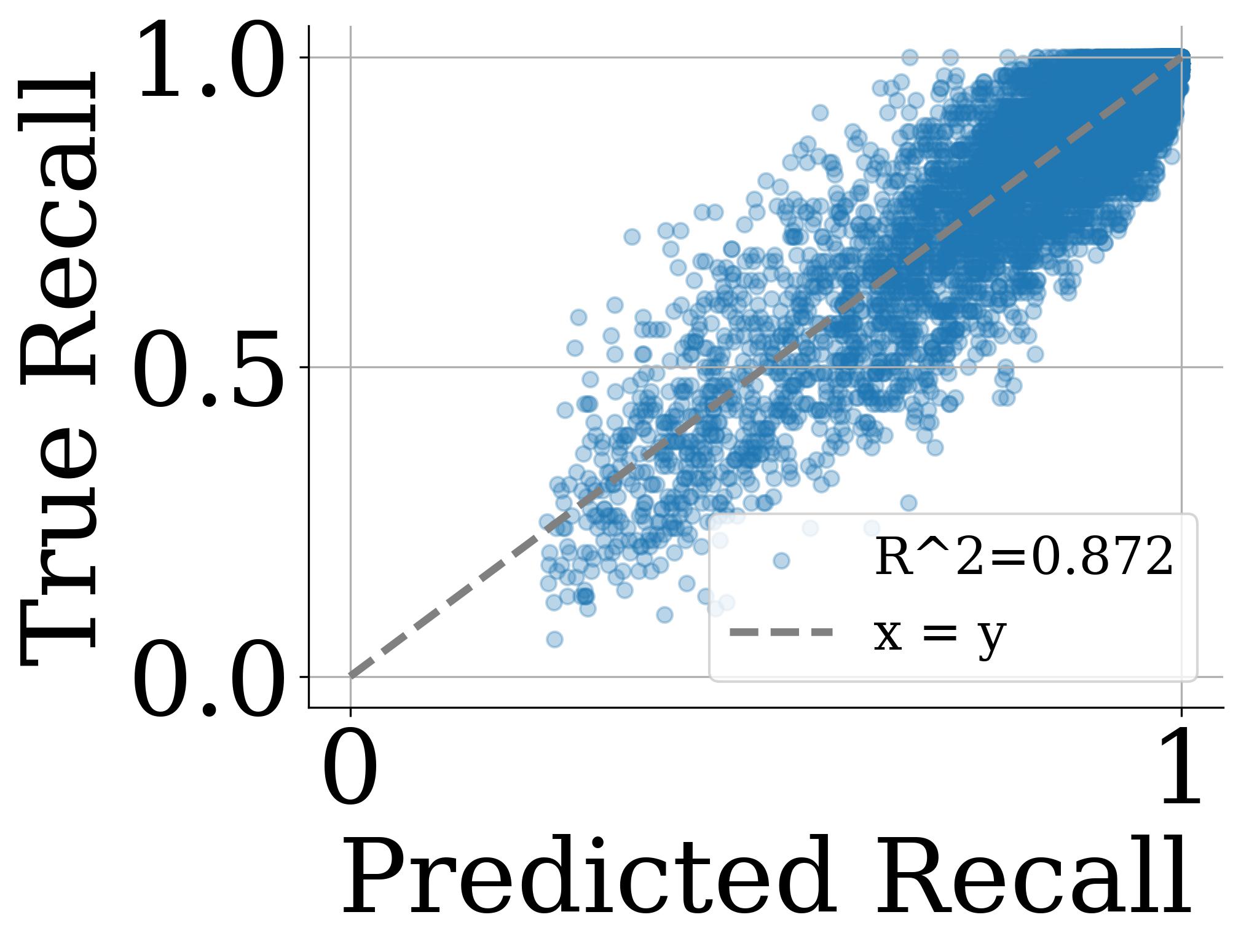}
    \caption{SIFT1M(G)}
\end{subfigure}
\begin{subfigure}{0.23\columnwidth}
    \includegraphics[width=\linewidth]{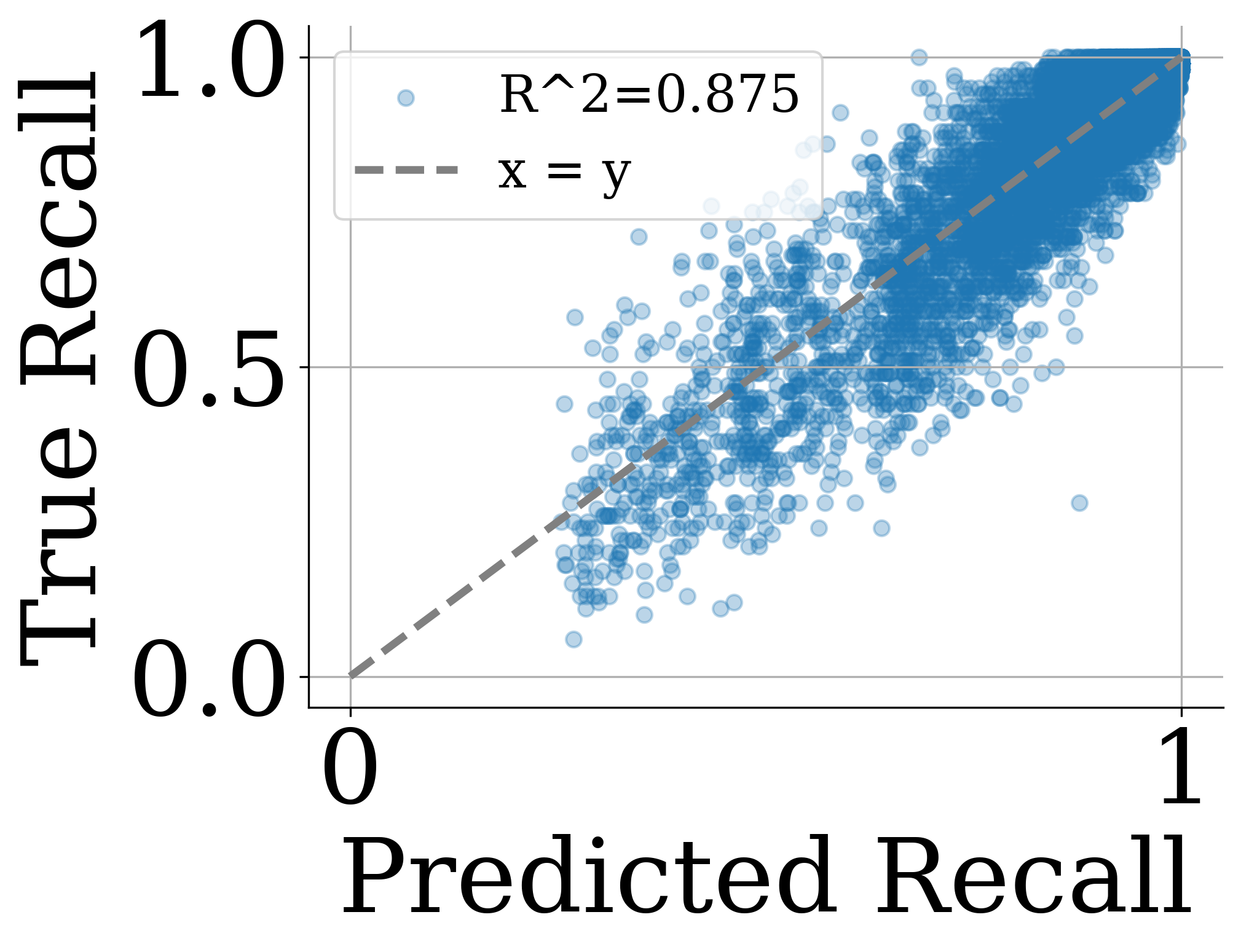}
    \caption{SIFT1M(D)}
\end{subfigure}
\begin{subfigure}{0.23\columnwidth}
    \includegraphics[width=\linewidth]{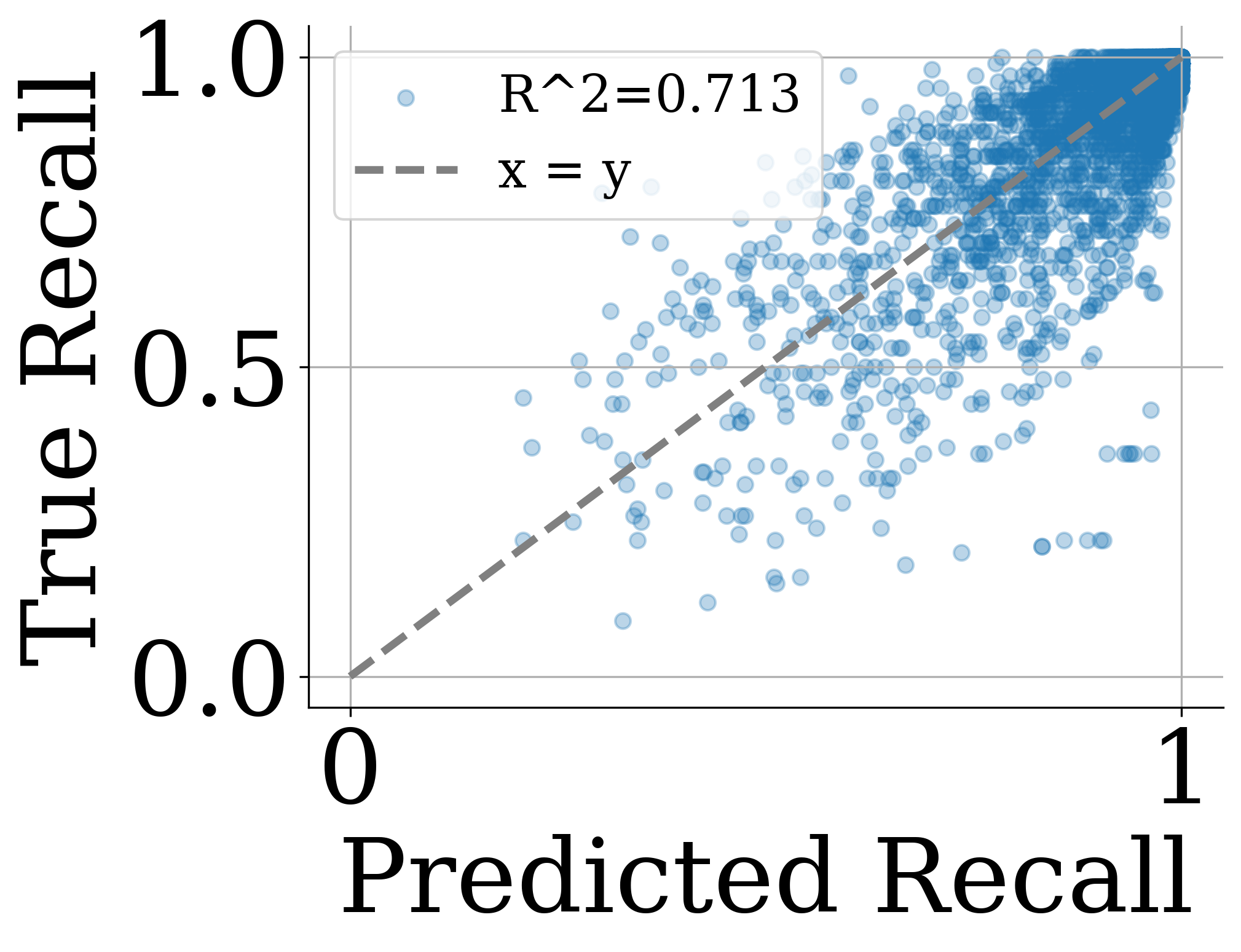}
    \caption{MNIST(G)}
\end{subfigure}
\begin{subfigure}{0.23\columnwidth}
    \includegraphics[width=\linewidth]{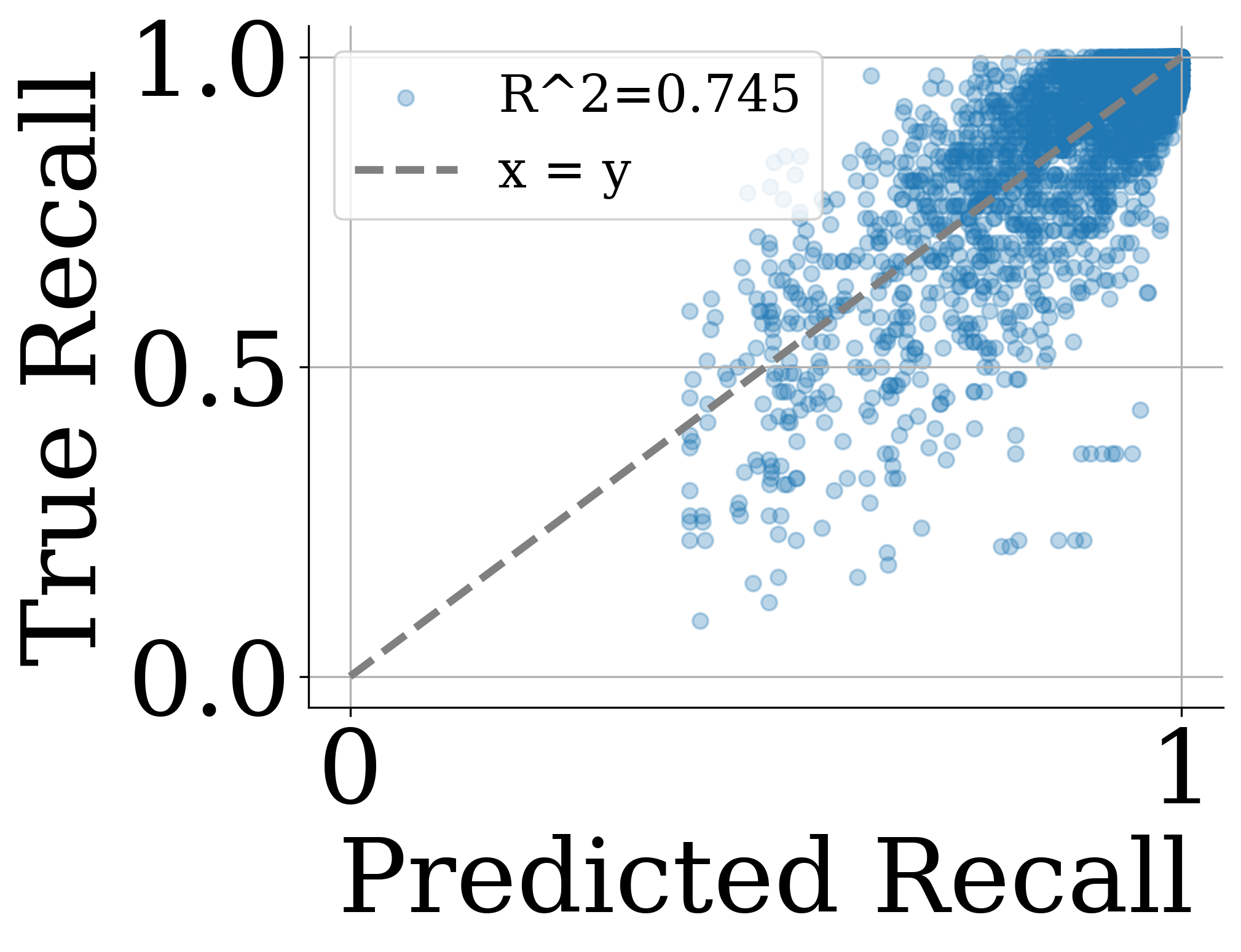}
    \caption{MNIST(D)}
\end{subfigure}
\begin{subfigure}{0.23\columnwidth}
    \includegraphics[width=\linewidth]{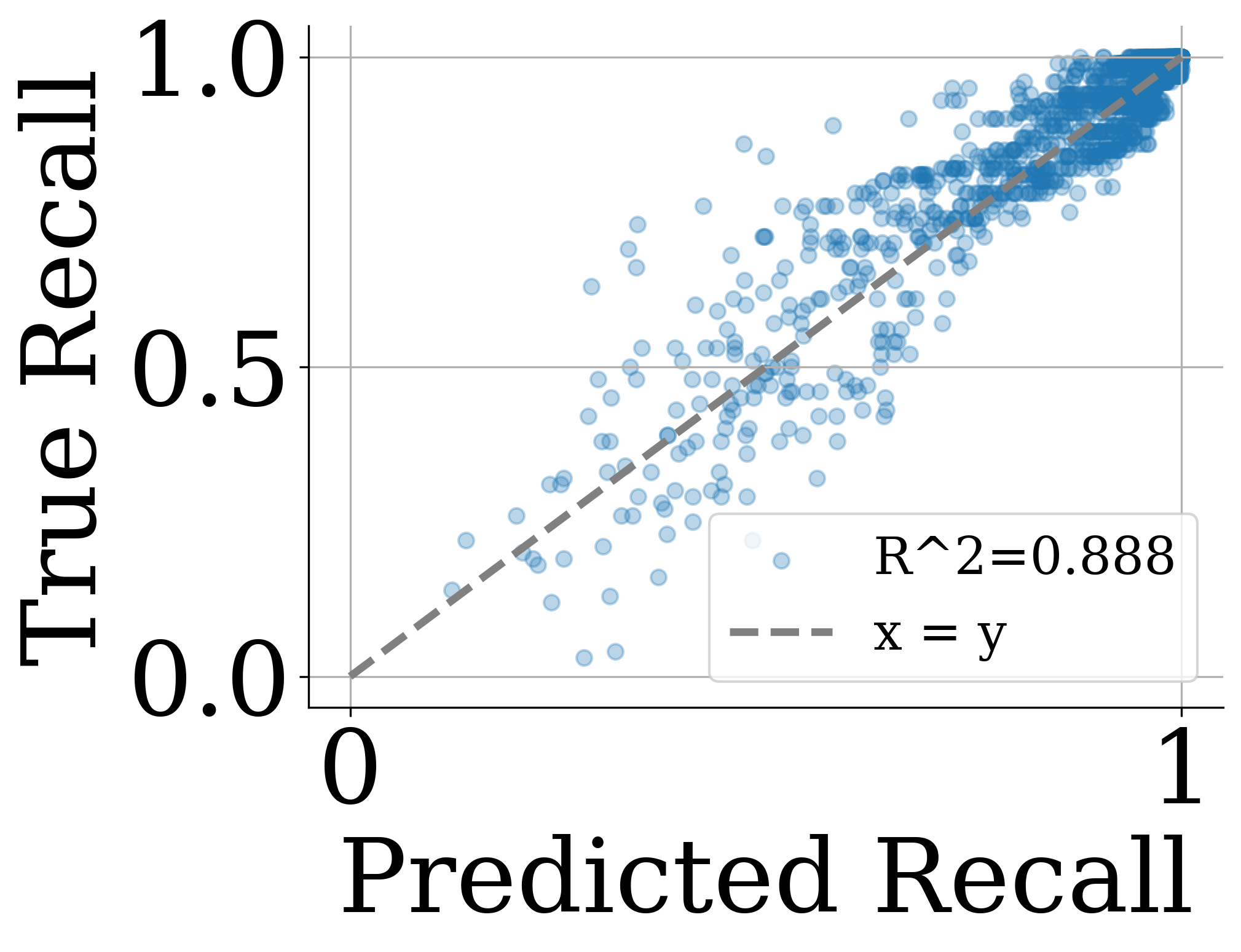}
    \caption{GIST(G)}
\end{subfigure}
\begin{subfigure}{0.23\columnwidth}
    \includegraphics[width=\linewidth]{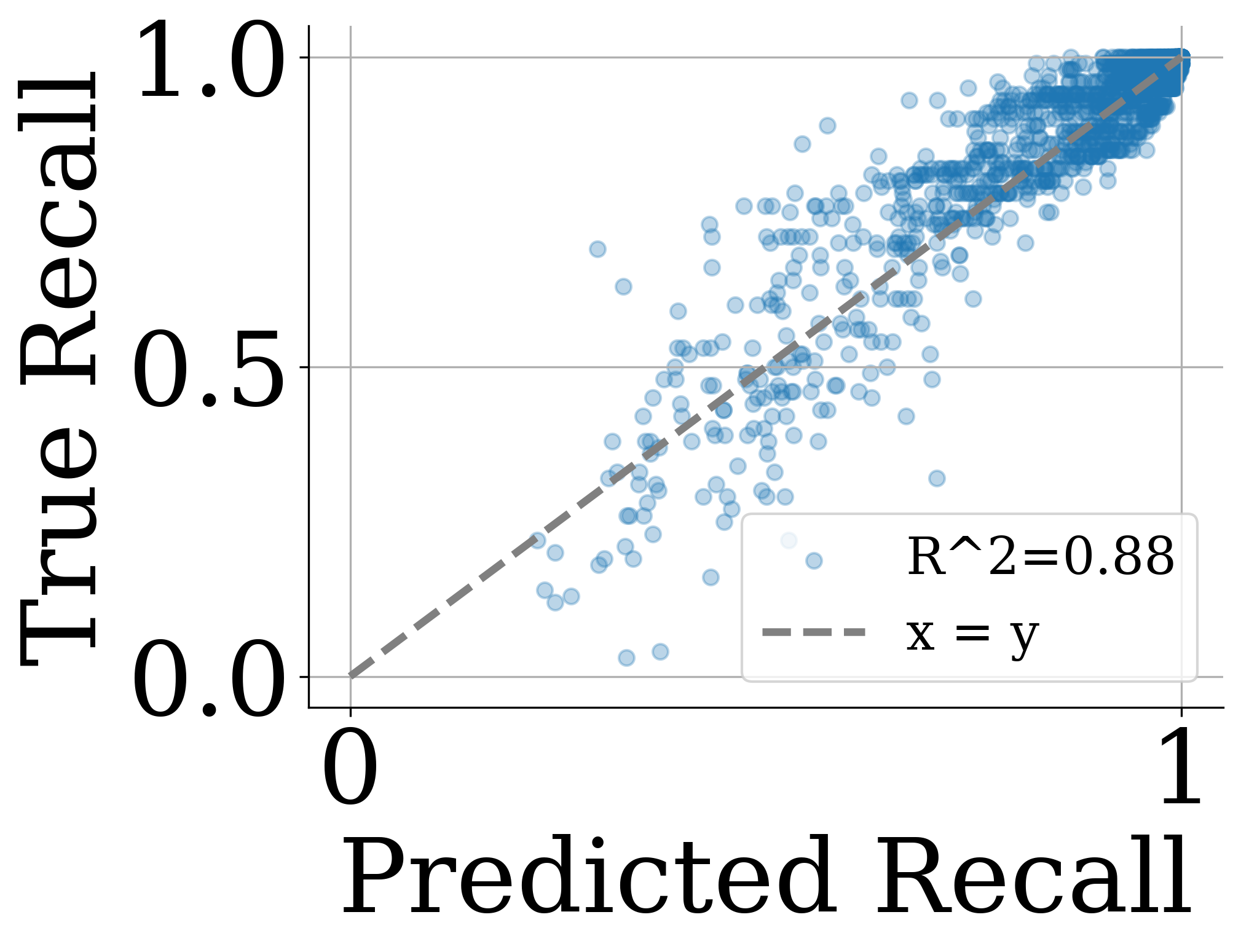}
    \caption{GIST(D)}
\end{subfigure}
\begin{subfigure}{0.23\columnwidth}
    \includegraphics[width=\linewidth]{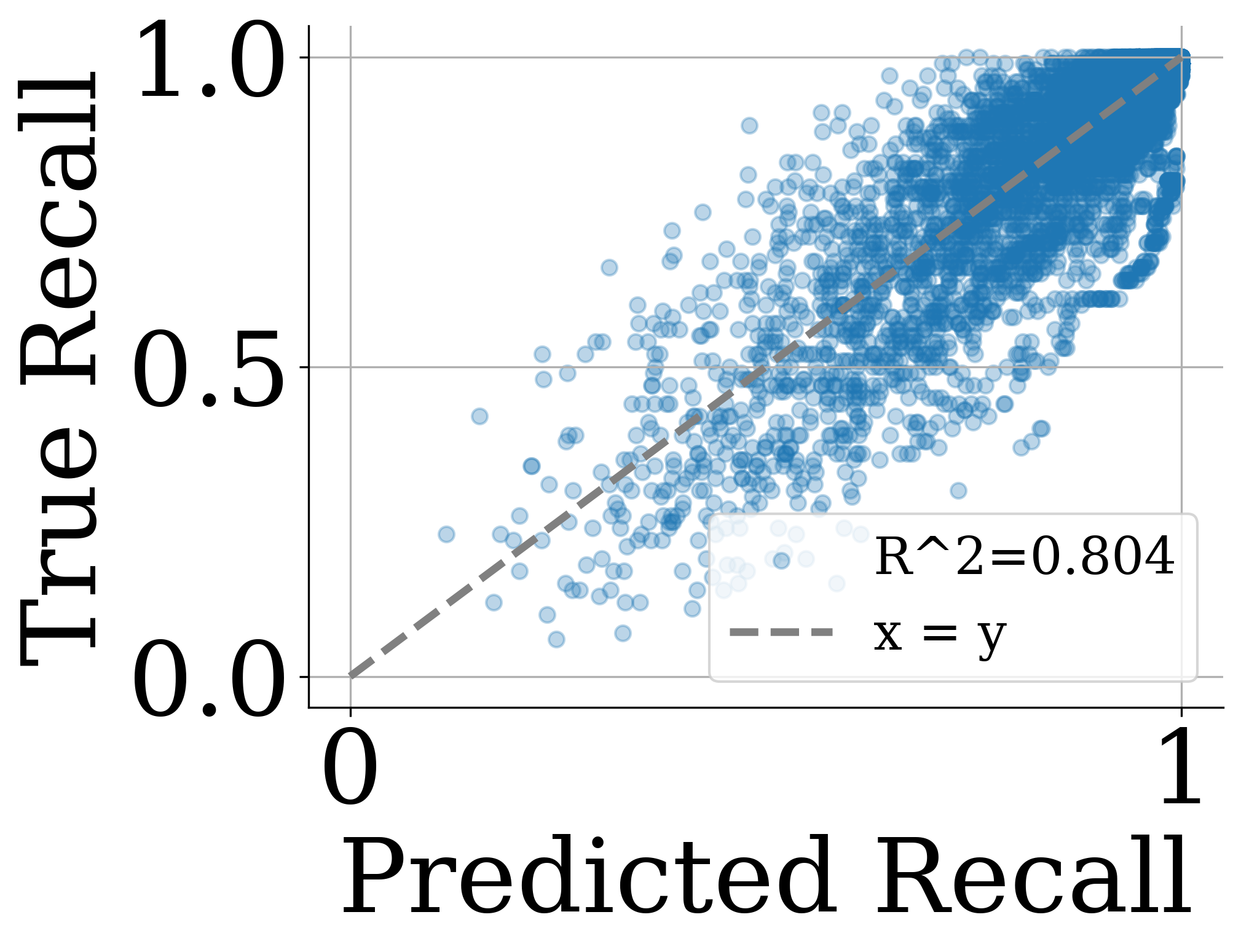}
    \caption{DEEP10M(G)}
\end{subfigure}
\begin{subfigure}{0.23\columnwidth}
    \includegraphics[width=\linewidth]{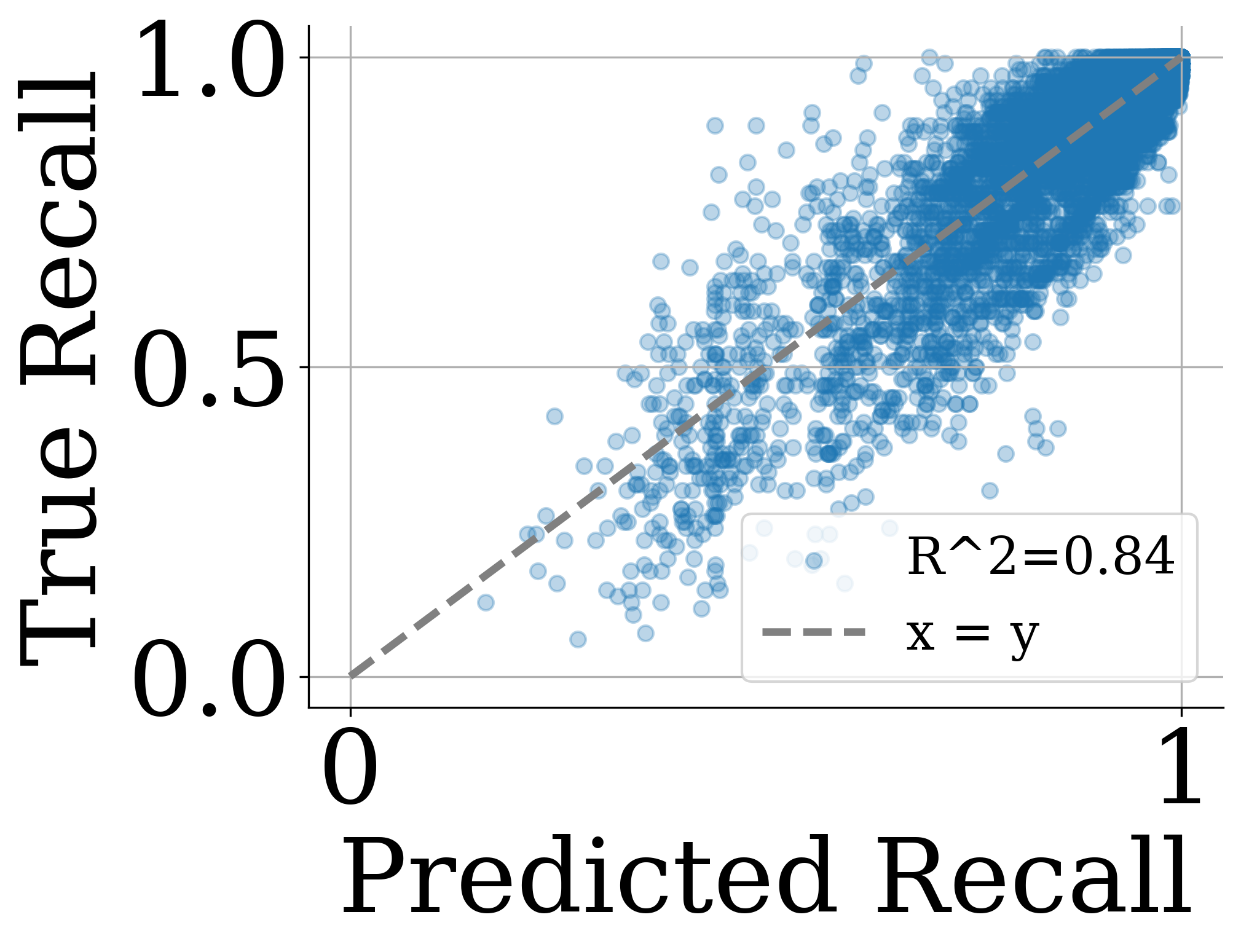}
    \caption{DEEP10M(D)}
\end{subfigure}
\caption{Prediction error analysis showing actual vs. predicted recall between GBDT (G) and DL (D) recall predictor models.}
\label{fig:prediction_error_combined}
\end{figure}

\subsection{Feature Studies} \label{sec:feature_studies}
We fix a set of nine features (Fig.\ref{fig:feature_importance}) and analyze their contributions using SHAP values\cite{shap} from the QASP-GBDT model. Cumulative cluster size and relative distance contribute most to positive recall predictions for both Euclidean and Angular datasets, though with high variance—expected since large cluster sizes and relative distances do not guarantee high recall for extremely hard queries. Feature ablation results (Table~\ref{tab:feature_ablation}) confirm that cumulative cluster size, relative distance, and local relative contrast are important across both dataset categories, while features such as jump in relative distance are more influential for angular datasets.

\begin{figure}[h]
\centering
    \includegraphics[width=\columnwidth]{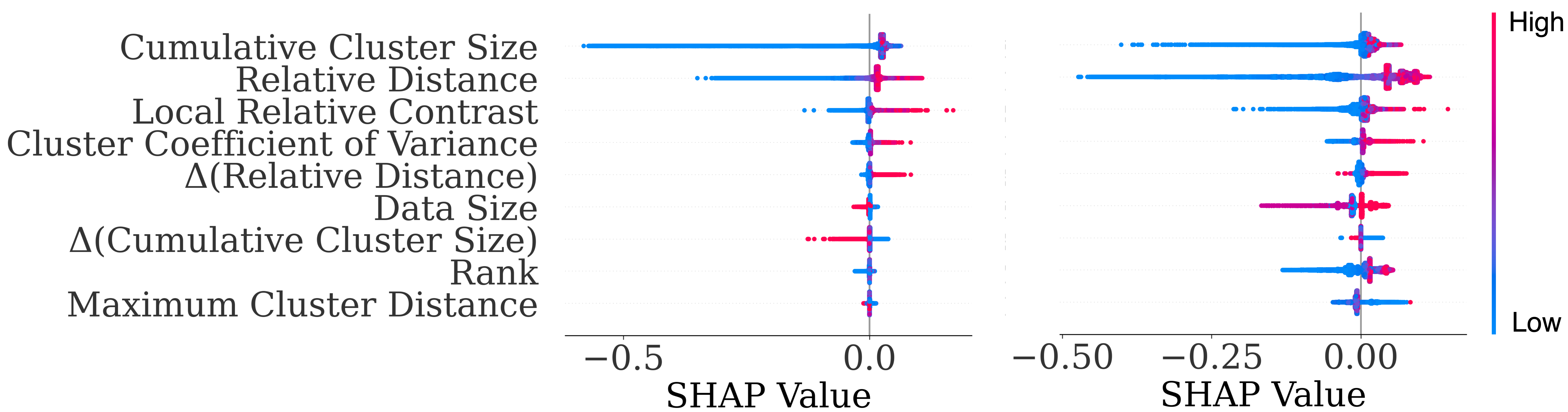}
\caption{Feature importance analysis. Left: Euclidean datasets. Right: Angular datasets.}
\label{fig:feature_importance}
\end{figure}

\begin{table}[t]
\centering
\caption{Comparison of model fit across all feature ablations for Euclidean and Angular datasets on validation set.}
\label{tab:feature_ablation}
\resizebox{\columnwidth}{!}{
\begin{tabular}{l|cc|cc}
\textbf{Feature Set} & \multicolumn{2}{c|}{\textbf{Euclidean Dataset}} & \multicolumn{2}{c}{\textbf{Angular Dataset}} \\
 & \textbf{MSE} $\downarrow$ & \textbf{R\textsuperscript{2}} $\uparrow$ & \textbf{MSE} $\downarrow$ & \textbf{R\textsuperscript{2}} $\uparrow$ \\
\hline
ALL & \textbf{0.00186} & \textbf{0.87123} & \textbf{0.00409} & \textbf{0.82741} \\
ALL \textbackslash \, Cumulative Cluster Size & 0.00193 & 0.86634 & 0.00417 & 0.82387 \\
ALL \textbackslash \, Relative Distance & 0.00246 & 0.82950 & 0.00454 & 0.80837 \\
ALL \textbackslash \, Local Relative Contrast & 0.00188 & 0.86954 & 0.00432 & 0.81743 \\
ALL \textbackslash \, Cluster Coeff. of Variance & 0.00192 & 0.86678 & 0.00413 & 0.82565 \\
ALL \textbackslash \, $\Delta$(\text{Relative Distance}) & 0.00191 & 0.86783 & 0.00419 & 0.82296 \\
ALL \textbackslash \,Data Size & 0.00188 & 0.86990 & 0.00415 & 0.82464 \\
ALL \textbackslash \, $\Delta$(\text{Cumulative Cluster Size}) & 0.00186 & 0.87096 & 0.00410 & 0.82686 \\
ALL \textbackslash \, Rank & 0.00186 & 0.87094 & 0.00416 & 0.82444 \\
ALL \textbackslash \, Maximum Cluster Distance & 0.00187 & 0.87078 & 0.00416 & 0.82407 \\
\end{tabular}
}
\end{table}

\section{Conclusion}
We introduced QASP to optimize vector search by predicting the complete recall progression curve per query via a single proactive inference, from which a search policy is derived for any recall target. QASP decouples the policy from specific targets or index configurations and enables domain adaptation with zero-shot or minimal fine-tuning. We provide theoretical guarantees, including that a finite sample suffices for convergence independent of dataset size and dimensionality, and a dominance condition where QASP's data access savings over fixed policies grow exponentially in intrinsic dimensionality. Our query variability-aware evaluation demonstrates the importance of minimizing recall variance across queries. QASP achieves significantly lower variance, lower deviation from target, and higher satisfaction rate while accessing similar or less data, with improvements most pronounced for hard queries and high recall regimes. QASP extends to hierarchical partitioning using inference-time scaling alone, achieving 99\% recall with 80\% less data access. QASP's progressive recall predictions further enable a lightweight reactive complement without additional inference.

\bibliographystyle{./IEEEtran}
\bibliography{references}

\section{AI Assistance Statement}
We used Anthropic's Claude model to assist with implementing the PCE-Net baseline from the original paper, as its code was not publicly available. Claude was also used to refactor the codebase for public release and to refine figure aesthetics and layout. All refactored code was manually verified against our original implementation. Experimental design, theoretical analysis, and scientific contributions are solely the authors' work.

\end{document}